\documentclass[floatfix,reprint,pra,aps,twocolumn,superscriptaddress,longbibliography, nobibnotes]{revtex4-2}
\usepackage[utf8]{inputenc}
\usepackage{float}
\usepackage{amsmath}

\usepackage{dcolumn}
\usepackage{bm}
\usepackage{amsmath,amssymb,amsthm,bm,amsfonts,bbm}

\newtheorem{lemma}{Lemma}

\usepackage{float}

\usepackage{graphicx}
\usepackage{dcolumn}
\usepackage{bm}
\renewcommand{\vec}[1]{\bm{#1}}
\usepackage{calrsfs}
\usepackage{booktabs}
\usepackage{wrapfig}
\usepackage{physics}
\usepackage{bbold}
\usepackage{quantikz}
\usepackage[colorlinks=true,linkcolor=blue,urlcolor=blue,citecolor=blue]{hyperref}
\usepackage{amsmath}
\usepackage{amssymb}
\usepackage[cal=cm]{mathalfa}
\usepackage{booktabs} 
\usepackage{multirow}
\usepackage{tikz}

\begin{document}
\setcounter{secnumdepth}{2}

\title{Testing and Diagnosing Parameter Miscalibration in Probabilistic Error Cancellation via Direct Fidelity Estimation}

\author{Francesc Sabater}
\email{francesc.sabater-garcia@tum.de}
\affiliation{Technical University of Munich, TUM School of Natural Sciences, Physics Department, 85748 Garching, Germany}
\affiliation{BMW Group, Munich, Germany}

\author{Tristan Kraft}
\affiliation{Technical University of Munich, TUM School of Natural Sciences, Physics Department, 85748 Garching, Germany}
\affiliation{Munich Center for Quantum Science and Technology (MCQST), Schellingstr. 4, 80799 München, Germany}

\author{Carlos A. Riofr\'io}
\affiliation{BMW Group, Munich, Germany}

\author{Barbara Kraus}
\affiliation{Technical University of Munich, TUM School of Natural Sciences, Physics Department, 85748 Garching, Germany}
\affiliation{Munich Center for Quantum Science and Technology (MCQST), Schellingstr. 4, 80799 München, Germany}

\date{\today}

\begin{abstract}
Probabilistic Error Cancellation (PEC) is a standard technique for mitigating noise in quantum computations, but its unbiasedness relies on accurate knowledge of the underlying noise parameters. In practice, these parameters are inferred from experimental data and may be inaccurate. We introduce a protocol that tests PEC parameter consistency on hardware by applying Direct Fidelity Estimation (DFE) to the effective mitigated state. The resulting DFE overlap serves not only as a witness of parameter miscalibration but, under appropriate conditions, also reveals whether the PEC parameters have been overestimated or underestimated. We establish this behavior for arbitrary circuits under global depolarizing noise and for Clifford circuits under general Pauli-stochastic noise. For arbitrary circuits subject to Pauli-stochastic noise, the underestimation result holds generally, while the corresponding overestimation result is established perturbatively for sufficiently small consistent miscalibrations under an additional nondegeneracy condition. We further show that the purity of the effective mitigated operator provides a more general diagnostic, allowing the miscalibration direction to be identified for arbitrary circuits and arbitrary parameter deviations. Finally, we describe a two-copy PEC procedure for measuring this quantity and show that, together with the DFE overlap, it determines the Hilbert--Schmidt distance from the ideal state. The sampling overhead incurred by our protocol is the same as that of PEC, and we discuss strategies to improve it.
\end{abstract}

\maketitle
\section{Introduction}
\label{sec: Introduction}

As quantum computing progresses toward the early fault-tolerant regime, noise remains a dominant challenge limiting the reliability of quantum computations~\cite{early_fault_tolerant,eisert2025mindgapsfraughtroad}. Current quantum processors remain susceptible to both coherent and incoherent errors arising from imperfect gate operations and environmental decoherence~\cite{Preskill2018quantumcomputingin,nisq_algorithms}. Accurately characterizing these errors is essential both for certifying quantum computations and for developing effective quantum error mitigation (QEM) strategies~\cite{QEM_review}. Recent works have further highlighted the importance of validating the noise models underlying error-mitigation procedures, including approaches to validate error-mitigated observable estimates~\cite{Barron2026} and cryptographically verifiable implementations of probabilistic error cancellation in delegated-computation settings~\cite{Yang2026}.

Among the most powerful noise-aware QEM techniques is Probabilistic Error Cancellation (PEC)~\cite{Temme_2017,Endo_2018,van_den_Berg_2023}, which stochastically inverts the noise channel at each gate to recover an unbiased estimate of the ideal expectation value of some observable. Unlike methods that do not require an explicit characterization of the gate-level noise, such as Zero Noise Extrapolation~\cite{Temme_2017} and data-driven methods based on Clifford or near-Clifford training circuits~\cite{Czarnik_2021,Strikis2021,Lowe2021}, the unbiasedness guarantee of PEC relies on accurate knowledge of the underlying noise model. In practice, however, the noise parameters associated with each gate are inferred from experimental data and may be biased, either through finite-sample uncertainty or through incorrect modeling assumptions such as Pauli-stochastic noise, neglected non-Markovian correlations, locality violations, or state-preparation and measurement (SPAM) errors~\cite{bounding_PEC_error}. This leads to a miscalibration of the PEC procedure. The resulting bias in the mitigated estimator has been analyzed and bounded in Refs.~\cite{bounding_PEC_error,Jin_2025}. A complementary question has so far remained largely unexplored: \textit{how can one efficiently test, directly on hardware, whether the noise parameters used for PEC are sufficiently accurate, and, if they are not, determine the direction of their miscalibration?}

Existing verification approaches address related but distinct problems. Techniques such as gate-set tomography~\cite{Nielsen2021gatesettomography} and Pauli channel learning~\cite{Harper2020Efficient,Chen2023LearnabilityPauliNoise,FlammiaWallman2020} characterize noise at the level of individual gates, but do not directly validate the end-to-end mitigation procedure. The protocol of Ref.~\cite{carrasco2024} verifies consistency between the measurement outcome distributions of the noisy and ideal circuits. Here, we pursue a complementary objective: an end-to-end test of the PEC procedure based on the overlap between the ideal state and the effective mitigated state. This approach is not restricted to a single fixed measurement distribution and can additionally provide information about the direction of parameter miscalibration.

In this work, we show that Direct Fidelity Estimation (DFE)~\cite{Flammia_2011,daSilva2011}, applied to the effective mitigated state, provides such a diagnostic. The central observation is that under the conditions specified below the DFE overlap, $\mathrm{Tr}(\rho \widetilde{\rho})$, between the ideal state $\rho$ and the effective mitigated state $\widetilde{\rho}$, which need not be positive semidefinite and is therefore, strictly speaking, an operator rather than a physical state, serves as a diagnostic for parameter miscalibration: a deviation of the DFE overlap from unity signals miscalibration of the error parameters. For simplicity, we therefore refer to $\widetilde{\rho}$ as the effective mitigated state in the following. Our key diagnostic result goes beyond mere testing: the \textit{direction} of the DFE overlap deviation from unity can reveal whether the error parameters have been overestimated or underestimated. This establishes a route from an experimentally measurable quantity to a diagnostic of PEC parameter miscalibration. Complementing the DFE-overlap result, we show that the purity of the effective mitigated state provides the full directional response for arbitrary circuits and arbitrary parameter deviations under consistent miscalibration. We also explain how this purity can be measured through pairs of physical PEC trajectories and how, together with the DFE overlap, it determines the Hilbert--Schmidt distance from the ideal state. An overview of the analytical results and their diagnostic interpretation is provided in Tables~\ref{tab:over_under_summary_theoretical}
and~\ref{tab:over_under_summary_experimental}.

The paper is structured as follows. Section~\ref{sec: overview of results} summarizes the protocol and main analytical results. Section~\ref{sec: preliminaries} introduces the PEC and DFE frameworks and establishes the notation used throughout. Section~\ref{sec: validation protocol} develops the experimental protocol: it constructs classically simulable surrogate circuits, derives the PEC-mitigated DFE-overlap estimator, introduces the two-copy purity estimator, and combines them into a  distance certificate. It also discusses strategies for reducing their PEC sampling overhead. Section~\ref{sec: calibrating} establishes the conditions under which the measured overlap and purity diagnose the direction of parameter miscalibration. Finally, Sec.~\ref{sec: applications} demonstrates the protocol numerically for Matchgate circuits with combined crosstalk and over-rotation errors and for global depolarizing noise, which are relevant noise models across different experimental platforms ~\cite{Haeffner2008,Boixo2018,Urbanek2021}.

\begin{figure}[t]
\centering
\includegraphics[width=\linewidth]{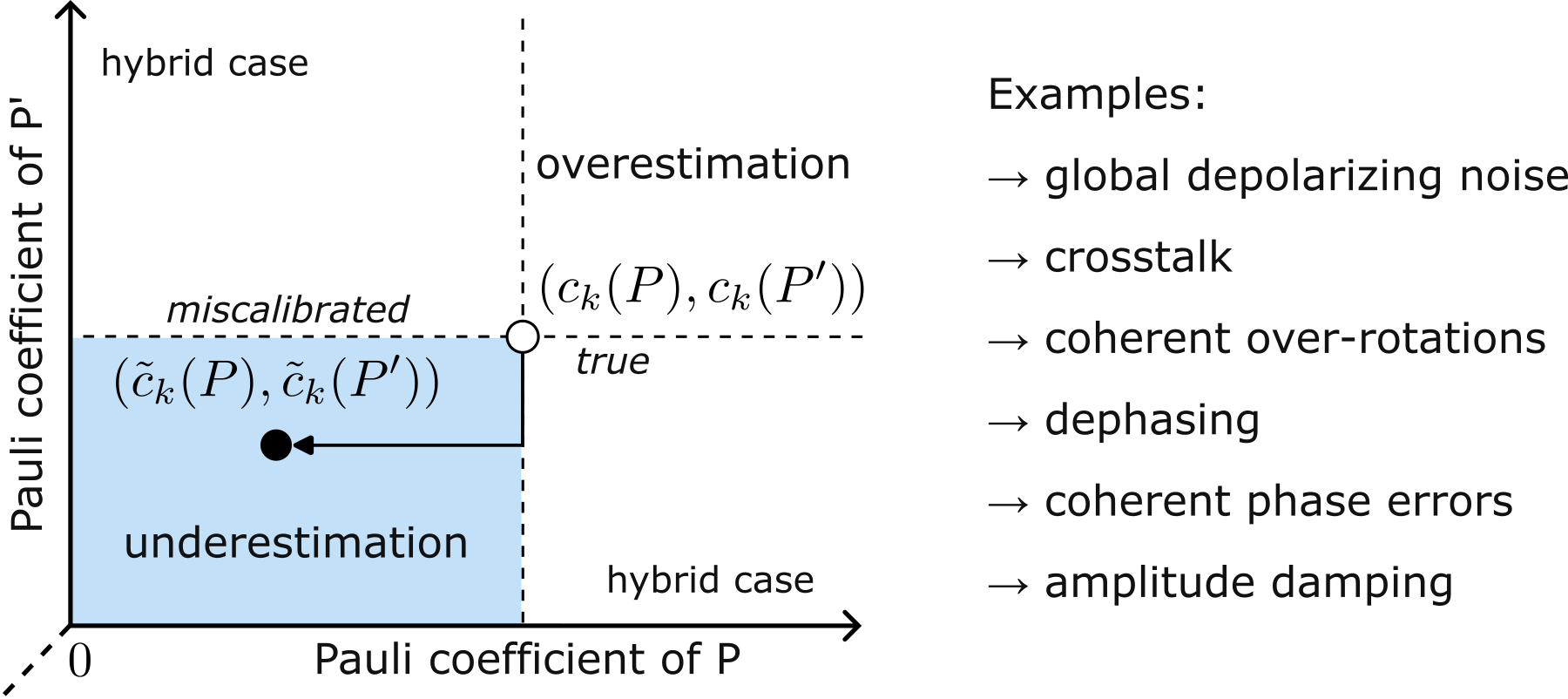}
\caption{
Illustration of over- and underestimation for a single error source. For a stochastic Pauli error channel, see Eq.~\eqref{eq: stochastic pauli channel}, underestimation is defined by $\widetilde c_k(P)\leq c_k(P)$ for all nonidentity Pauli operators $P$, with at least one strict inequality, while overestimation is defined by the corresponding reversed inequalities. For the physical error models considered in this work, the miscalibration of their corresponding physical parameters implies these coefficientwise conditions, either exactly or, for the coherent-error models, within the stated perturbative regime.
When several error sources are concatenated, we apply these definitions separately to each constituent source. We call the resulting miscalibration \emph{consistent} if all constituent sources are underestimated or all are overestimated according to the coefficientwise definition above; if different sources are miscalibrated in opposite directions, we refer to the case as \emph{hybrid} miscalibration. Importantly, coefficientwise ordering is not, in general, preserved for the Pauli coefficients of the concatenated channel. Instead, the constituent-level conditions imply the corresponding inequalities for the effective Pauli fidelities (see main text). These effective Pauli-fidelity conditions are preserved under concatenation and are the relevant quantities for the analytical results below.
}
\label{fig:over_under_definition}
\end{figure}

\section{Overview of results}
\label{sec: overview of results}

The central quantity in our analysis is the DFE overlap
\begin{equation}
Y \equiv \operatorname{Tr}(\rho\widetilde{\rho}),
\end{equation}
where $\rho=\ketbra{\psi}$ is the ideal target state and $\widetilde{\rho}$ is the effective operator produced by PEC. Despite the fact that $\widetilde{\rho}$ need not be positive semidefinite under parameter miscalibration, we sometimes also call it effective state. Note that, in case $\widetilde{\rho}$ is not a valid state, $Y$ is not necessarily a fidelity and may lie outside $[0,1]$. Experimentally, $Y$ is estimated by the DFE estimator $\widehat Y$. The role of the protocol is therefore to use the observed value of $\widehat Y$ to test whether the effective mitigated result is consistent with ideal error cancellation (Sec.~\ref{sec: validation protocol}) and, under additional assumptions, to diagnose the direction of parameter miscalibration (Sec.~\ref{sec: calibrating}).

For universal target circuits, the ideal Pauli expectation values required to estimate $Y$ are generally not classically accessible. The protocol therefore uses a classically simulable surrogate chosen to preserve as closely as possible the structure and hardware-relevant features of the target circuit. PEC is executed for this surrogate on hardware, and $Y$ is estimated for the resulting effective mitigated state. The surrogate constructions and the circuit features they retain are discussed in Sec.~\ref{sec: surrogate}.

For a single error source we define over- and underestimation at the level of the Pauli coefficients, as described below and illustrated in Fig.~\ref{fig:over_under_definition}. We formulate these definitions for stochastic Pauli channels, see Eq.~\eqref{eq: stochastic pauli channel}, since Pauli twirling provides a natural reduction of a general noise channel to its Pauli component. For many physically relevant error models, miscalibration of the underlying physical parameter---the depolarizing rate being one example---induces these coefficientwise conditions, either exactly or in a perturbative regime.

Throughout the directionality analysis, we restrict to the regime in which the Pauli fidelities of the true and estimated noise channels are strictly positive. A simple sufficient condition is that the identity coefficients of both channels are larger than $1/2$, as shown in Appendix~\ref{appendix inverse channel}. When several error sources are present, this assumption is imposed on every constituent source. We call the miscalibration \emph{consistent} if all constituent sources are underestimated or all are overestimated according to the coefficientwise definition. These constituent-level conditions then imply
\begin{equation}
0<f_{k,P,\mathrm{eff}}\leq 1
\end{equation}
for consistent underestimation, and
\begin{equation}
f_{k,P,\mathrm{eff}}\geq 1
\end{equation}
for consistent overestimation, where
\begin{equation}
f_{k,P,\mathrm{eff}}
=
\frac{f_{k,P}}{\widetilde f_{k,P}}.
\end{equation}
Here, $f_{k,P}$ and $\widetilde f_{k,P}$ are the diagonal PTM entries of the true noise channel $\mathcal D_k$ and the estimated noise channel $\widetilde{\mathcal D}_k$, respectively; the inverse map $\widetilde{\mathcal D}_k^{-1}$ is the map used in PEC. These effective Pauli-fidelity conditions, unlike the coefficientwise ordering of the constituent channels, are preserved under concatenation. Cases in which different error sources are miscalibrated in opposite directions are referred to as hybrid miscalibration.

Our main analytical results establish implications from these over- and underestimation regimes to the theoretical overlap $Y$. For arbitrary circuits under global depolarizing noise, and more generally for Clifford circuits under stochastic Pauli noise, consistent underestimation implies $Y\leq1$, while consistent overestimation implies $Y\geq1$. The former result extends to arbitrary circuits under stochastic Pauli noise, while the latter persists perturbatively beyond Clifford circuits under the conditions specified below.

Another central quantity is the DFE purity
\begin{equation}
    \mathcal{P}
    \equiv
    \operatorname{Tr}(\widetilde{\rho}^{2}),
\end{equation}
which, in contrast to the purity of a quantum state, can exceed one when $\widetilde{\rho}$ is not positive semidefinite. Analogously to the DFE overlap, we will refer to it as the ``DFE purity'' in the following. The DFE purity provides a complementary and more general directional result: for arbitrary circuits and arbitrary parameter deviations, consistent underestimation implies $\operatorname{Tr}(\widetilde{\rho}^{2})\leq1$, whereas consistent overestimation implies $\operatorname{Tr}(\widetilde{\rho}^{2})\geq1$. These results are summarized in Table~\ref{tab:over_under_summary_theoretical}.

When hybrid miscalibration is excluded, the corresponding inequalities provide experimental diagnostics through both $\widehat{Y}$ and the DFE-purity estimator $\widehat{\mathcal{P}}$. A statistically significant value of either estimator above unity rules out consistent underestimation, whereas a value below unity rules out
consistent overestimation, subject to the assumptions of the corresponding analytical result. In particular, the DFE-purity diagnostic applies to arbitrary circuits and arbitrary parameter deviations, while the DFE-overlap diagnostic in the arbitrary-circuit overestimation regime requires the perturbative conditions specified below. The finite-sample diagnostics based on $\widehat{Y}$ and $\widehat{\mathcal{P}}$ are summarized in
Table~\ref{tab:over_under_summary_experimental}.

\begin{table*}[t]
\centering
\renewcommand{\arraystretch}{1.5}
\setlength{\tabcolsep}{14pt}

\begin{tabular*}{\textwidth}{@{\extracolsep{\fill}}cccc}
\hline

\textbf{Noise model}
& \textbf{Circuit class}
& \textbf{Miscalibration scenario}
& \textbf{Response} \\

\hline

\multirow{2}{*}{Global depolarizing noise}
& \multirow{2}{*}{Arbitrary circuit}
& Consistent overestimation
& $Y\geq 1$ \\

&
& Consistent underestimation
& $Y\leq 1$ \\

\hline

\multirow{8}{*}{Pauli-stochastic noise}
& \multirow{2}{*}{\shortstack[c]{One-layer arbitrary\\circuit}}
& Consistent overestimation
& $Y\geq 1$ \\

&
& Consistent underestimation
& $Y\leq 1$ \\

\cline{2-4}

& \multirow{2}{*}{Clifford circuit}
& Consistent overestimation
& $Y\geq 1$ \\

&
& Consistent underestimation
& $Y\leq 1$ \\

\cline{2-4}
\noalign{\vskip 0.5ex}

& \multirow{2}{*}{Arbitrary circuit}
& \shortstack[c]{Consistent overestimation\\
with small deviations}
& \shortstack[c]{$Y\geq 1$\\
(perturbative result only)} \\

&
& Consistent underestimation
& $Y\leq 1$ \\

\cline{2-4}

& \multirow{2}{*}{Arbitrary circuit}
& Consistent overestimation
& $\operatorname{Tr}(\widetilde{\rho}^{2})\geq 1$ \\

&
& Consistent underestimation
& $\operatorname{Tr}(\widetilde{\rho}^{2})\leq 1$ \\

\hline
\end{tabular*}

\caption{
Summary of the DFE-overlap and DFE-purity responses. For a concatenated noise model, we call the miscalibration consistent when all constituent channels are miscalibrated in the same direction according to the single-channel definitions of Sec.~\ref{sec: directionality general}.
The DFE-overlap bound for arbitrary circuits in the overestimation regime is only valid for sufficiently small deviations, and additionally requires the first-order coefficient $A$ in Eq.~\eqref{eq: perturbative DFE coefficient} to be strictly positive. In contrast, the corresponding DFE-purity bounds and the arbitrary-circuit underestimation bound do not require this perturbative restriction. Hybrid cases, in which different constituent channels are miscalibrated in opposite directions, are not covered by the analytical results. For the precise single-channel definitions of underestimation and overestimation, see Sec.~\ref{sec: directionality general} or Fig.~\ref{fig:over_under_definition}.
}
\label{tab:over_under_summary_theoretical}
\end{table*}

\begin{table*}[t]
\centering
\renewcommand{\arraystretch}{1.45}
\setlength{\tabcolsep}{10pt}

\begin{tabular*}{\textwidth}{@{\extracolsep{\fill}}cccc}
\hline
\noalign{\vskip 0.8ex}

\textbf{Noise model}
& \textbf{Circuit class}
& \textbf{Measurement}
& \textbf{\shortstack[c]{Diagnostic\\interpretation}} \\

\noalign{\vskip 0.5ex}
\hline

\multirow{2}{*}{Global depolarizing noise}
& \multirow{2}{*}{Arbitrary circuit}
& $\widehat{Y}>1+2\epsilon$
& Consistent overestimation \\

&
& $\widehat{Y}<1-2\epsilon$
& Consistent underestimation \\

\hline

\multirow{8}{*}{Pauli-stochastic noise}
& \multirow{2}{*}{\shortstack[c]{One-layer arbitrary\\circuit}}
& $\widehat{Y}>1+2\epsilon$
& Consistent overestimation \\

&
& $\widehat{Y}<1-2\epsilon$
& Consistent underestimation \\

\cline{2-4}

& \multirow{2}{*}{Clifford circuit}
& $\widehat{Y}>1+2\epsilon$
& Consistent overestimation \\

&
& $\widehat{Y}<1-2\epsilon$
& Consistent underestimation \\

\cline{2-4}

& \multirow{2}{*}{Arbitrary circuit}
& $\widehat{Y}>1+2\epsilon$
& Consistent overestimation \\

&
& $\widehat{Y}<1-2\epsilon$
& \shortstack[c]{Consistent underestimation expected\\
in the small-deviation regime} \\

\cline{2-4}

& \multirow{2}{*}{\shortstack[c]{Arbitrary circuit\\
(purity diagnostic)}}
& $\widehat{\mathcal{P}}>1+\epsilon_{\mathrm{pur}}$
& Consistent overestimation \\

&
& $\widehat{\mathcal{P}}<1-\epsilon_{\mathrm{pur}}$
& Consistent underestimation \\

\hline
\end{tabular*}

\caption{
Summary of the experimentally accessible DFE-overlap and
DFE-purity diagnostics. The DFE-overlap estimator
$\widehat{Y}$ satisfies
$\operatorname{Pr}[|\widehat{Y}
-\operatorname{Tr}(\widetilde{\rho}\rho)|
\leq2\epsilon]\geq1-2\delta$, whereas the DFE-purity estimator
$\widehat{\mathcal{P}}$ satisfies
$\operatorname{Pr}[|\widehat{\mathcal{P}}
-\operatorname{Tr}(\widetilde{\rho}^{2})|
\leq\epsilon_{\mathrm{pur}}]
\geq1-\delta_{\mathrm{pur}}$ when the number of paired
measurements satisfies
Eq.~\eqref{eq:purity_pec_samples}. Consequently, the
thresholds displayed in the table establish, at the
corresponding confidence level, that the exact overlap or
purity lies above or below one. The table assumes that the
possible deviations are restricted to consistent
overestimation and consistent underestimation; hybrid
miscalibration is not covered. For stochastic Pauli noise
and arbitrary circuits, interpreting
$\widehat{Y}<1-2\epsilon$ as evidence of underestimation
relies on the perturbative small-deviation result for
overestimation. In contrast, the DFE-purity diagnostic applies
to arbitrary circuits and arbitrary parameter deviations
under the stated effective-Pauli-fidelity conditions.
}
\label{tab:over_under_summary_experimental}
\end{table*}

\section{Preliminaries}
\label{sec: preliminaries}
In this section we introduce the technical ingredients used throughout the paper. We first review probabilistic error cancellation (PEC) and introduce the notation for the true and estimated noise parameters, including the effective mitigated state resulting from their mismatch. We then recall direct fidelity estimation (DFE) and the sampling estimator used in the remainder of the paper.

\subsection{Imperfect Probabilistic Error Cancellation}
\label{sec: PEC}

\begin{figure*}[t]
    \centering
    \includegraphics[width=\linewidth]{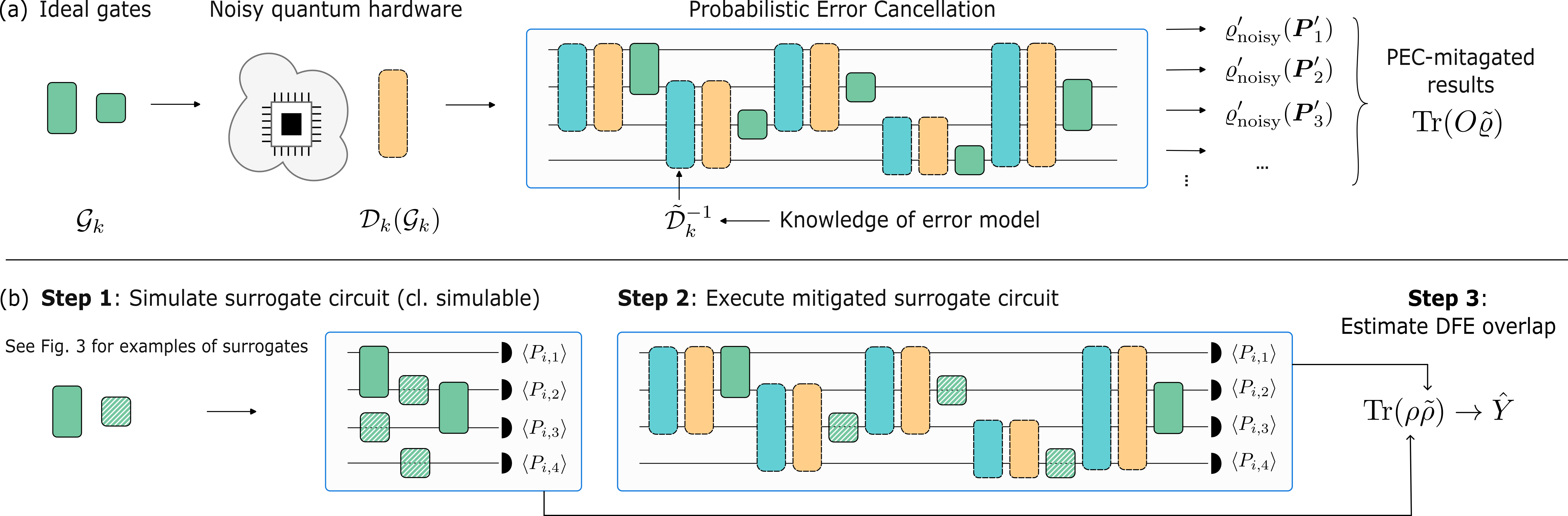}
    \caption{
    (a) The standard PEC procedure. Given a gate-level noise model $\mathcal{D}_k(\mathcal{G}_k)$ for the noisy hardware, PEC probabilistically applies a quasiprobability decomposition of the inverse noise map, such that averaging over the resulting circuits recovers the ideal expectation value when the noise model is correctly specified~\cite{Temme_2017,van_den_Berg_2023}. If the noise model parameters are miscalibrated, the resulting PEC estimator is generally biased~\cite{Jin_2025}.
    (b) Overview of the protocol for estimating the DFE overlap. Our protocol proceeds in three steps. Step 1: The target circuit is mapped to a classically simulable surrogate. Step 2: The surrogate circuit is executed on noisy hardware using PEC. Step 3: DFE is used to estimate the DFE overlap $\operatorname{Tr}(\rho\widetilde{\rho})\mapsto \widehat Y$ between the ideal surrogate state and the effective PEC-mitigated state by sampling Pauli observables. The resulting estimate can then be used to witness miscalibration of the error parameters as described in Sec.~\ref{sec: calibrating} and summarized in Table~\ref{tab:over_under_summary_experimental}.
    }
\label{fig:overview}
\end{figure*}

Probabilistic Error Cancellation is a quantum error mitigation technique to estimate the noiseless expectation value of an observable by stochastically inverting the noise channel associated with each noisy gate~\cite{Temme_2017}, see Fig.~\ref{fig:overview}(a) for an illustration. The central idea is to expand noise-free operations as linear combinations of implementable noisy operations, allowing one to statistically undo the effect of noise across multiple circuit executions. To achieve this inversion, it is essential to know both the parameters and the noise model characterizing each noisy operation---that is, one must have a detailed and accurate description of how noise acts on the system.

Throughout this section, and throughout the rest of the manuscript, we focus on gate noise and do not explicitly model state-preparation and measurement (SPAM) errors. These effects can be treated separately using dedicated calibration and mitigation methods; in particular, readout errors can be tackled using measurement-error mitigation~\cite{Bravyi2021Measurement}.

Consider an ideal circuit $U$ that we aim to implement, consisting of $s$ gates denoted by $G_k$,
\begin{equation}
    U = \prod_{k=1}^s G_k.
\end{equation}
Due to hardware noise, we do not have access to the noiseless operation $\mathcal{G}_k(\cdot) = G_k (\cdot) G_k^\dagger$, but rather to a noisy version of it. We assume that the noise affecting each gate can be described by a gate-dependent stochastic Pauli channel that acts on a finite number of qubits $n$. Such an assumption is justified, since any error channel can be converted into a Pauli stochastic one through Pauli twirling~\cite{van_den_Berg_2023,Wallman_2016}. The noisy implementation of the gate is then given by
\begin{align}
\label{eq: stochastic pauli channel}
\mathcal{D}_k(\mathcal{G}_k)(\cdot)
&= \sum_{P \in \mathcal{P}_n} c_k(P)\, P G_k (\cdot) G_k^\dagger P, 
\end{align}
where the sum runs over all Pauli operators of the Pauli group acting on $n$ qubits, $\mathcal{P}_n = \{I, X, Y, Z\}^{\otimes n}$. The $4^n$ coefficients $c_k(P) \ge 0$ satisfy $\sum_{P \in \mathcal{P}_n} c_k(P) = 1$, and they fully characterize the noise channel associated with gate $k$.

By inserting the inverse of the noise channel $\mathcal{D}_k^{-1}$, one can recover the ideal operation $\mathcal{G}_k$, 
\begin{align}\label{eq:idelG}
\mathcal{G}_k(\cdot) 
&= \mathcal{D}_k^{-1} \mathcal{D}_k(\mathcal{G}_k)(\cdot) =\sum_{P' \in \mathcal{P}_n}a_k(P')\mathcal{D}_k(\mathcal{P'}\mathcal{G}_k)(\cdot) \notag\\
&=\!\!\!\sum_{P' \in \mathcal{P}_n}\! \sum_{P \in \mathcal{P}_n} 
\!a_k(P')\, c_k(P)\!\, P' P G_k(\cdot) G_k^\dagger P P', 
\end{align}
where we assume that we can implement Pauli gates without modifying the original noise channel coefficients $c_k(P)$. This can be interpreted as a linear combination weighted by the coefficients $a_k(P')$ of noisy operations $\mathcal{D}_k(\mathcal{P'}\mathcal{G}_k)$ that one can implement on the noisy hardware through $\mathcal{P'}\mathcal{G}_k$.  Here, the inverse noise map is itself a Pauli map with coefficients $a_k(P')$. These coefficients depend on the original parameters $c_k(P)$ and can be determined by solving the system of equations defined by
\begin{equation}
    \label{system equations inverse coefficient}
    \mathcal{I} = \mathcal{D}_k^{-1} \mathcal{D}_k=\sum_{P' \in \mathcal{P}_n} \sum_{P \in \mathcal{P}_n} a_k(P')\, c_k(P)\, \mathcal{P'}\mathcal{P},
\end{equation}
which yields $4^n$ equations that determine the inverse coefficients $a_k(P')$. More generally, the inverse map exists if and only if none of the Pauli fidelities $f_{k,P}$ vanishes. The inverse map $\mathcal{D}_k^{-1}$ exists if and only if
\begin{equation}
    f_{k,P}\neq 0
    \qquad
    \text{for every } P\in\mathcal{P}_n.
\end{equation}
Appendix~\ref{appendix inverse channel} shows that the stronger condition $c_k(I^{\otimes n})>1/2$ is sufficient, since it guarantees $f_{k,P}>0$ for every $P$ and therefore also places the channel in the positive-fidelity regime used in our directionality analysis. The coefficients of the inverse map satisfy $\sum_{P'\in\mathcal{P}_n}a_k(P')=1$, but need not all be nonnegative. Consequently, the inverse map is generally not completely positive and does not represent a physical quantum channel. In PEC, the implementable operations are instead sampled with probabilities proportional to $|a_k(P')|$, while the corresponding signs are incorporated into the estimator, as detailed below.

In practice, the noise parameters $c_k(P)$ are unknown and need to be reconstructed as finite-sample estimates from experimental data, for instance through process tomography. Therefore, the noise parameters used for the PEC procedure $\widetilde{c}_k(P)$, can be \emph{miscalibrated}: they carry statistical fluctuations and, as mentioned in the introduction, can additionally be biased by several factors. 

When PEC is applied using the miscalibrated parameters $\widetilde{c}_k(P)$, one effectively introduces an inverse map $\widetilde{\mathcal{D}}_k^{-1}$ characterized by coefficients $\widetilde{a}_k(P')$. Consequently, PEC performed under miscalibrated coefficients does not recover the ideal operation as in Eq.~\eqref{eq:idelG}, but instead produces an effective residual map
\begin{align}
\widetilde{\mathcal{G}}_k(\cdot) 
&= \widetilde{\mathcal{D}}_k^{-1} \mathcal{D}_k(\mathcal{G}_k)(\cdot)
= \mathcal{D}_{k,\mathrm{eff}}(\mathcal{G}_k)(\cdot) \notag\\
&= \! \! \!\sum_{P' \in \mathcal{P}_n} \! \sum_{P \in \mathcal{P}_n} 
\widetilde{a}_k(P')\, c_k(P) \!\, P' P G_k(\cdot) G_k^\dagger P P'.
\end{align}
Importantly, this map, resulting from an incorrect inversion, is not guaranteed to be completely positive.

When applying PEC, one effectively inverts the noise affecting each gate $k$. After inverting the estimated noise channel at each gate, the resulting effective operation is given by
\begin{equation}
\label{eq: final unitary}
\widetilde{\mathcal{U}}(\cdot) = \prod_{k=1}^s \widetilde{\mathcal{G}}_k(\cdot)
= \sum_{\vec{P}'} \widetilde{a}(\vec{P}')\, \mathcal{U}'_{\text{noisy}}(\vec{P}')(\cdot),
\end{equation}
where $\widetilde{a}(\vec{P}') = \widetilde{a}_1(P'_1)\widetilde{a}_2(P'_2)\dots\widetilde{a}_s(P'_s)$, and the sum runs over all possible combinations of Pauli operators $\vec{P}'$ across all gates. Consequently, the summation contains exponentially many terms in the number of gates $s$. Here, $\mathcal{U}'_{\text{noisy}}(\vec{P}')$ denotes the noisy implementation of the modified circuit $U$ with the inserted Paulis $\vec{P}'$ between layers, i.e.,
\begin{equation}
    \mathcal{U}'_{\text{noisy}}(\vec{P}')(\cdot) = \prod_{k=1}^s \mathcal{D}_k(\mathcal{P}'_k \mathcal{G}_k)(\cdot).
\end{equation}
To work with the resulting operator, or state, we use Eq.~(\ref{eq: final unitary}) to write the resulting quantum state after applying each operation as
\begin{equation}
    \widetilde{\rho} = \sum_{\vec{P}'} \widetilde{a}(\vec{P}')\, \rho'_{\text{noisy}}(\vec{P}').
\end{equation}
Thus, PEC defines a (not necessarily preparable) effective operator $\widetilde{\rho}$ as a linear combination of experimentally preparable noisy states, $\rho'_{\text{noisy}}(\vec{P}')$. If the noise parameters are correct, $\widetilde{c}_k(P)=c_k(P)$ for all $k,P$, then $\widetilde{\rho}=\rho$, where $\rho$ is the ideal output state. If the noise parameters are miscalibrated, then generally $\widetilde{\rho}\neq\rho$, and $\widetilde{\rho}$ may even fail to be positive semidefinite because the linear combination may contain negative coefficients.

The effective state $\widetilde{\rho}$ is not physically accessible, since it is in general not possible to experimentally prepare a linear combination of quantum states. Nevertheless, one can always access the expectation value of any observable of interest $O$ through the noisy expectation values of the modified circuits,
\begin{equation}
\label{eq: O PEC}
    \mathrm{Tr}[O\widetilde{\rho}] = \sum_{\vec{P}'} \widetilde{a}(\vec{P}')\, \mathrm{Tr}[O\rho'_{\text{noisy}}(\vec{P}')].
\end{equation}
The central idea of PEC is to recast this sum, which contains exponentially many terms, so that it can be evaluated through sampling~\cite{Temme_2017}. To this end, the magnitudes of the coefficients are turned into a probability distribution, and their signs are taken care of separately. Concretely, we introduce probabilities $\widetilde{p}(\vec{P}') = |\widetilde{a}(\vec{P}')|/\widetilde{Q}$, where
\begin{equation}
\label{eq: Q}
\widetilde{Q} = \sum_{\vec{P}'} |\widetilde{a}(\vec{P}')| 
= \sum_{P_1'} |\widetilde{a}_1(P_1')| \dots \sum_{P_s'} |\widetilde{a}_s(P_s')|.
\end{equation}
With this, Eq.~(\ref{eq: O PEC}) can be rewritten as
\begin{equation}
\label{eq: O PEC prob}
    \mathrm{Tr}[O\widetilde{\rho}] 
    = \widetilde{Q} \sum_{\vec{P}'} \widetilde{p}(\vec{P}')\, 
    \mathrm{sign}[\widetilde{a}(\vec{P}')]\,
    \mathrm{Tr}[O\rho'_{\text{noisy}}(\vec{P}')],
\end{equation}
where $\widetilde{p}(\vec{P}')$ is now a probability distribution, whose probabilities are proportional to the magnitudes $|\widetilde{a}(\vec{P}')|$, and the sign is taken care of independently.
Therefore, one obtains the PEC estimator by independently sampling one modified circuit $\vec{P}'_j\sim\widetilde{p}$ on each shot, executing it once, and measuring the observable $O$~\cite{Temme_2017,Endo_2018}. Denoting the resulting single-shot measurement outcome by $\mu_j(O)$, the PEC estimator is
\begin{equation}
    \label{eq: OPEC}
    \widetilde{O}_{\mathrm{PEC}}
    =
    \frac{\widetilde{Q}}{N_{\mathrm{shots}}}
    \sum_{j=1}^{N_{\mathrm{shots}}}
    \operatorname{sign}
    \!\left[
        \widetilde{a}(\vec{P}'_j)
    \right]
    \mu_j(O),
    \qquad
    \vec{P}'_j\sim\widetilde{p}.
\end{equation}
Here, each sampled modified circuit $\vec{P}'_j$ is executed only once, and $\mu_j(O)$ denotes the outcome of the corresponding single-shot measurement of $O$. In the case of correct error parameters, we recover the same expression without tildes, and the mitigated estimator is unbiased~\cite{Temme_2017}:
\begin{equation}
    \lim_{N_{\text{shots}} \to \infty} O_{\mathrm{PEC}} = \mathrm{Tr}[O\rho].
\end{equation}
However, when using miscalibrated error parameters, the estimator exhibits a finite bias~\cite{bounding_PEC_error}. For an observable with bounded outcomes, the PEC estimator has variance of order $\mathcal{O}(\widetilde{Q}^{\,2})$. Therefore, achieving a root-mean-square statistical error of at most $\xi$ requires
\begin{equation}
    N_{\mathrm{shots}}
    =
    \mathcal{O}\!\left(
        \frac{\widetilde{Q}^{\,2}}{\xi^2}
    \right).
\end{equation}
Hence, $\widetilde{Q}^{\,2}$ represents the sampling overhead relative to the unmitigated estimator.

From Eq.~\eqref{eq: Q}, it follows that the sampling overhead scales exponentially with the number of gates. This result was already established for the PEC method~\cite{Temme_2017,Endo_2018} and is consistent with more general findings showing that any QEM method incurs an exponential overhead with circuit size in the worst case~\cite{boundEM,Quek_2024}. Nevertheless, the gate error rate appears as a prefactor in the exponent, making QEM methods as well as PEC feasible for circuit sizes on the order of the inverse gate error~\cite{zimboras2025mythsquantumcomputationfault,aharonov2025importanceerrormitigationquantum}.

\subsection{Direct Fidelity Estimation}
\label{sec: DFE}
DFE is a protocol designed to estimate the fidelity between a pure, known quantum state---in our case, $\rho=\ketbra{\psi}$, obtained by applying the ideal quantum circuit $U$ to a given pure input state---and an unknown, possibly mixed quantum state~\cite{Flammia_2011,daSilva2011}. In this work, we want to apply the DFE protocol to the effective mitigated state $\widetilde{\rho}$, which may not correspond to a valid physical quantum state. The DFE procedure remains applicable since it only requires access to the expectation values of observables with respect to $\widetilde{\rho}$, rather than direct access to $\widetilde{\rho}$ itself. It does, however, require the ideal Pauli expectation values $\mathrm{Tr}[P\rho]=\bra{\psi}P\ket{\psi}$ as classical side information, since these define both the sampling distribution and the estimator below. Later, we will therefore utilize surrogate circuits for which $\rho$ remains classically tractable, such that the required Pauli expectation values can be computed or estimated classically efficiently (e.g., a Clifford, or Matchgate circuit).

The DFE overlap between $\rho$ and $\widetilde{\rho}$ can be written as  $\mathrm{Tr}(\rho\widetilde{\rho})$, since $\rho$ is a pure state. The protocol is then based on expanding both states in the Pauli basis, e.g.,
\begin{equation}
    \rho = \sum_{P \in \mathcal{P}_N} \chi_{\rho}(P)\frac{P}{\sqrt{2^N}},
\end{equation}
where $N$ is the total number of qubits, and $\chi_{\rho}(P) = \mathrm{Tr}[P\rho]/\sqrt{2^N}$, and similarly for $\widetilde\rho$. The DFE overlap can then be expressed as
\begin{equation}
    \mathrm{Tr}(\rho\widetilde{\rho}) = \sum_{P \in \mathcal{P}_N} \chi_{\widetilde{\rho}}(P)\chi_{\rho}(P)
    = \sum_{P \in \mathcal{P}_N} \frac{\chi_{\widetilde{\rho}}(P)}{\chi_{\rho}(P)}\, \chi_{\rho}^2(P).
\end{equation}
A key element of the DFE method is the observation that $\chi_{\rho}^2(P)$ defines a discrete probability distribution over the Pauli group, since $\chi_{\rho}^2(P) \geq 0$ for all $P$, and $\mathrm{Tr}(\rho^2)=\sum_{P \in \mathcal{P}_N}\chi_{\rho}^2(P) = 1$, since $\rho$ is a pure state. Pauli operators for which $\chi_\rho(P)=0$ have zero sampling probability and are therefore omitted from the support of the distribution. One can then define the estimator
\begin{equation}
    \label{eq: Xp}
    X_P = \frac{\chi_{\widetilde{\rho}}(P)}{\chi_{\rho}(P)} 
    = \frac{\mathrm{Tr}[\widetilde{\rho} P]}{\mathrm{Tr}[\rho P]},
\end{equation}
for which it holds that $\mathbb{E}_{P \sim \chi_{\rho}^2(P)}[X_P]=\mathrm{Tr}(\widetilde{\rho}\rho)$~\cite{Flammia_2011,Kliesch_2021}. The DFE overlap can therefore be estimated by sampling Pauli operators according to $\chi_{\rho}^2(P)$ and estimating the corresponding mitigated expectation values, while the ideal expectation values are provided as classical side information. For generic states, the resulting measurement budget can scale as $2^N$ for fixed precision and confidence. Although this is exponentially better than the $4^N$ scaling of full state tomography, it may still be prohibitive for large systems. Precise bounds for the PEC-mitigated setting are derived in Sec.~\ref{sec: DFE_mitigated}.

This exponential dependence can be avoided for \emph{well-conditioned} states. A state $\rho$ is said to be well-conditioned with parameter $\alpha$ if, for every Pauli operator $P$, either $\mathrm{Tr}[\rho P] = 0$ or $|\mathrm{Tr}[\rho P]| \geq \alpha$~\cite{Flammia_2011}. Prominent examples include stabilizer states ($\alpha=1$) and Dicke states~\cite{Flammia_2011}. For such states, a measurement budget of $O(\log(\delta^{-1})/(\alpha^2\epsilon^2))$ suffices to estimate the DFE overlap to additive precision $2\epsilon$ with probability at least $1-2\delta$.

\section{Testing PEC parameter consistency}
\label{sec: validation protocol}

We present a protocol for testing the consistency of the error parameters used to implement PEC. The protocol consists of three main steps, summarized schematically in Fig.~\ref{fig:overview}(b).

\begin{itemize}
    \item[--] \textbf{Step 1}: Surrogate construction. The universal target circuit is replaced by a structurally similar, classically simulable surrogate circuit for which the ideal Pauli expectation values required by DFE can be efficiently obtained.

    \item[--] \textbf{Step 2}: Probabilistic error cancellation. PEC is applied to the surrogate circuit on noisy hardware, leading to the effective mitigated state $\widetilde{\rho}$.

    \item[--] \textbf{Step 3}: DFE-overlap estimation. DFE is used to estimate the overlap $Y=\operatorname{Tr}(\rho\widetilde{\rho})$, providing a test of the consistency of the error parameters used in PEC for the surrogate circuit executed on hardware.
\end{itemize}

We first describe the construction of classically simulable surrogate circuits in Sec.~\ref{sec: surrogate}. We then construct the PEC-mitigated DFE-overlap estimator, derive its sampling bounds, and formulate the corresponding statistical consistency test in Sec.~\ref{sec: DFE_mitigated}. Because unit overlap alone does not certify recovery of the ideal state, Sec.~\ref{sec: purity certification} introduces an unbiased estimator of $\operatorname{Tr}(\widetilde{\rho}^{2})$ and combines it with the DFE overlap to obtain a finite-confidence estimate of the Hilbert--Schmidt distance from the ideal state. Finally, Sec.~\ref{sec: reducing PEC overhead} discusses strategies for reducing the PEC normalization and hence the sampling costs of both estimators.

\subsection{Classically simulable surrogate circuits}
\label{sec: surrogate}
Both standard DFE and its application to the effective mitigated state require access to the ideal Pauli expectation values $\operatorname{Tr}[P\rho]$, which are generally inefficient to compute for universal circuits. We therefore replace the target circuit $U$ by a classically simulable surrogate designed to retain as closely as possible its circuit structure and hardware implementation while making the quantities required by DFE efficiently accessible. The protocol then tests the PEC procedure on this closely related surrogate circuit.

We consider three surrogate constructions, summarized in Fig.~\ref{fig:circuit_transforms}. First, a Clifford surrogate can be obtained by compiling the circuit such that non-Clifford gates are single-qubit $R_Z$ rotations and replacing their angles by Clifford angles~\cite{Czarnik_2021}. Second, the circuit can be mapped to a Matchgate (MG) circuit by doubling the number of qubits and replacing non-MG two-qubit gates, such as CZ gates, by MG gates such as fermionic SWAPs~\cite{carrasco2024}. Both constructions allow strong classical simulation and hence direct computation of the ideal Pauli expectation values required by DFE. Third, following Ref.~\cite{carrasco2024}, one can retain the same gadget implementation of the encoded circuit while replacing the magic input states by modified states $|M'\rangle$, yielding an efficiently weakly simulable circuit. Here, the ideal Pauli expectation values may need to be estimated from samples, introducing additional statistical uncertainty. Using multiple surrogate constructions that preserve different features of the target computation can provide additional evidence that the observed behavior is not specific to a particular simplification.

For Clifford surrogates generated from stabilizer inputs, the DFE relevance distribution $\chi_\rho(P)^2$ can be sampled efficiently by sampling uniformly from the stabilizer group. For MG surrogates initialized in pure fermionic Gaussian states, including computational-basis states, the output remains a pure fermionic Gaussian state. The Majorana-sampling algorithm of Ref.~\cite{Collura_2026} can then sample the corresponding Pauli strings in polynomial time. Thus, for the Clifford and MG surrogates considered here, both the required Pauli expectation values and the relevance distribution are efficiently accessible.

Classically simulable Clifford and near-Clifford circuits have also been used as training data in learning-based QEM and Clifford data regression-type methods~\cite{Strikis2021,Lowe2021}. There, they are used to learn mitigation parameters or mappings between noisy and ideal expectation values. In contrast, our surrogates are not used to construct the mitigation; they provide classically tractable test cases for evaluating an independently specified PEC procedure through DFE.

\begin{figure}[t]
    \centering
    \includegraphics[width=0.95\linewidth]{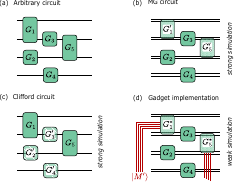}
    \caption{
    Schematic overview of the three structurally similar but classically efficiently simulable surrogate circuits for an arbitrary circuit of interest ({\bf a}). 
    ({\bf b}) The original circuit is transformed into an MG circuit by doubling the number of qubits and replacing the non-MG two-qubit gates by MGs, such as fermionic SWAPs~\cite{carrasco2024}. ({\bf c}) The single-qubit rotations are modified so that the resulting circuit is a Clifford circuit, while preserving the original circuit structure and connectivity~\cite{Czarnik_2021}. Both the Clifford and MG constructions allow \emph{strong} classical simulation, i.e. efficient classical computation of outcome-probabilities~\cite{carrasco2024}, and therefore provide efficient access to the ideal Pauli expectation values required by DFE. ({\bf d}) Gadget implementation of the original encoded circuit, denoted by $*$. To obtain a weakly simulable circuit---one from which the output can be sampled classically and efficiently---the magic input states are modified. In this case, the ideal Pauli expectation values required by DFE may have to be estimated from samples, introducing an additional statistical uncertainty in the DFE-overlap estimate.  For more details on both the gadget implementation of the original two-qubit gates and the modified input state, we refer the reader to Ref.~\cite{carrasco2024}. These transformations guarantee classical tractability, but not equality of the hardware noise affecting the surrogate and target circuits; transferring the test result to the target circuit relies on the noise-preservation assumption discussed in the main text.}
    \label{fig:circuit_transforms}
\end{figure}

\subsection{Estimating the DFE overlap of the mitigated state}
\label{sec: DFE_mitigated}
Unlike standard DFE, applying DFE to the effective mitigated state $\widetilde{\rho}$ requires the Pauli expectation values to be estimated using PEC. This introduces an additional sampling overhead, which we quantify below. To estimate $\operatorname{Tr}[\widetilde{\rho}P]$ in Eq.~\eqref{eq: Xp}, we use the PEC estimator $\widetilde{P}_{\mathrm{PEC}}$ [cf.~Eq.~\eqref{eq: OPEC}] to construct an estimator $\widehat{X}$ and estimate the DFE overlap from $\ell$ sampled Pauli operators,
\begin{equation}
    \widehat{Y} = \frac{1}{\ell}\sum_{i=1}^{\ell}\widehat{X}_{P_i}.
\end{equation}
The resulting estimator has two sources of statistical uncertainty: the finite number $m_P$ of measurement samples used to estimate each $\widetilde{P}_{\mathrm{PEC}}$ and the finite number $\ell$ of sampled Pauli operators. We closely follow the DFE analysis of Ref.~\cite{Kliesch_2021} and derive the corresponding sampling bounds for the mitigated case in Appendix~\ref{appendix DFE bounds}.

In particular, to guarantee
\begin{equation}
\label{main bound precision Y}
    \Pr\!\left[
    \left|\widehat{Y}-\operatorname{Tr}(\widetilde{\rho}\rho)\right|
    \leq 2\epsilon
    \right]
    \geq 1-2\delta,
\end{equation}
it is sufficient to choose
\begin{equation}
    \ell =
    \left\lceil
    \frac{\widetilde{Q}^2}{\epsilon^2\delta}
    \right\rceil,
    \hspace{25pt}
    m_P =
    \left\lceil
    \frac{2\widetilde{Q}^2}
    {2^N\chi_\rho(P)^2\ell\epsilon^2}
    \ln\!\left(\frac{2}{\delta}\right)
    \right\rceil.
\end{equation}
The resulting expected measurement budget is
\begin{equation}
\label{eq: sampling complexity DFE mitigated}
    N_t =\mathbb{E}\!\left[\sum_{i=1}^{\ell}m_{P_i}\right]
    \leq
    1+
    \frac{\widetilde{Q}^2}{\epsilon^2\delta}
    +
    \frac{2\cdot2^N\widetilde{Q}^2}{\epsilon^2}
    \ln\!\left(\frac{2}{\delta}\right).
\end{equation}
For comparison, standard DFE between valid quantum states has expected measurement budget
\begin{equation}
\label{eq: sampling complexity DFE}
    N_t
    \leq
    1+\frac{1}{\epsilon^2\delta}
    +\frac{2\cdot2^N}{\epsilon^2}
    \ln\!\left(\frac{2}{\delta}\right).
\end{equation}
Thus, applying DFE to the effective mitigated state introduces an additional sampling overhead of approximately $\widetilde{Q}^2$, matching the scaling of the PEC overhead incurred when estimating an observable~\cite{Temme_2017}.

For non-well-conditioned states, the DFE sampling complexity in Eqs.~\eqref{eq: sampling complexity DFE mitigated} and \eqref{eq: sampling complexity DFE} scales as $2^N$, even when the underlying circuit is Clifford or MG. This is nevertheless exponentially smaller than the $4^N$ scaling required for full quantum state tomography of $\widetilde{\rho}$~\cite{Cramer_2010}. The exponential DFE scaling can be avoided by using a well-conditioned test state (see Sec.~\ref{sec: DFE}). As shown in the well-conditioned-state analysis of
Appendix~\ref{appendix DFE bounds}, to guarantee Eq.~\eqref{main bound precision Y} it is sufficient to choose
\begin{equation}
    \ell =
    \left\lceil
    \frac{2\widetilde{Q}^2\ln(2/\delta)}
    {\epsilon^2\alpha^2}
    \right\rceil,
    \hspace{25pt}
    m_P=1,
\end{equation}
which gives the expected measurement budget
\begin{equation}
    N_t =
    \left\lceil
    \frac{2\widetilde{Q}^2\ln(2/\delta)}
    {\epsilon^2\alpha^2}
    \right\rceil
    \leq
    1+
    \frac{2\widetilde{Q}^2}
    {\alpha^2\epsilon^2}
    \ln\!\left(\frac{2}{\delta}\right).
\end{equation}
Thus, for a well-conditioned test state with constant $\alpha$, the sampling complexity is independent of the number of qubits.

A simple choice of such a test is to apply the protocol to the circuit $UU^\dagger$ acting on the all-zero input state, rather than to $U$ alone. The ideal output is then the all-zero state, for which $\alpha=1$ and the required Pauli expectation values can be computed efficiently. This construction therefore allows the protocol to be applied even when $U$ is neither Clifford nor MG. A limitation is that errors in the forward and inverse circuits may partially cancel, which can reduce the sensitivity of the test to some noise processes.

An alternative is to use an ensemble of randomized mirror circuits, which can probe different circuit depths, qubit subsets, and noise features~\cite{Proctor_2022,mayer2023theorymirrorbenchmarkingdemonstration}. Their ideal outputs are computational-basis states and hence stabilizer states with $\alpha=1$, so the well-conditioned DFE bounds remain applicable. Such circuits can complement the $UU^\dagger$ construction by reducing the possibility that the observed behavior is specific to a single test circuit.

With this, we can estimate the DFE overlap between the effective mitigated state $\widetilde{\rho}$ and the ideal state $\rho$ for Clifford and MG surrogate circuits, as well as for the simple mirror construction $UU^\dagger$ associated with an arbitrary circuit $U$. The estimated overlap $\widehat{Y}$ provides a statistical consistency test for the PEC procedure. From Eq.~\eqref{main bound precision Y},
\begin{equation}
\label{eq: overlap concentration}
    \Pr\!\left[
    \left|\widehat{Y}-\operatorname{Tr}(\widetilde{\rho}\rho)\right|
    \leq 2\epsilon
    \right]
    \geq 1-2\delta.
\end{equation}
We therefore define
\begin{equation}
\label{eq: protocol}
\begin{cases}
\text{Pass the DFE-overlap test},
& \text{if } |\widehat{Y}-1|\leq 2\epsilon,\\[2pt]
\text{Fail the DFE-overlap test},
& \text{if } |\widehat{Y}-1|>2\epsilon.
\end{cases}
\end{equation}
The test is two-sided because $\widetilde{\rho}$ need not be positive semidefinite, and hence the DFE overlap may lie either above or below one. If the mitigation exactly recovers the ideal state, $\widetilde{\rho}=\rho$, then $\operatorname{Tr}(\widetilde{\rho}\rho)=1$ and Eq.~\eqref{eq: overlap concentration} guarantees that the overlap test is passed with probability at least $1-2\delta$. Conversely, if
\begin{equation}
    \left|1-\operatorname{Tr}(\widetilde{\rho}\rho)\right|>4\epsilon,
\end{equation}
then Eq.~\eqref{eq: overlap concentration} implies that the overlap test is failed with probability at least $1-2\delta$. The interval $|1-\operatorname{Tr}(\widetilde{\rho}\rho)|\leq4\epsilon$ therefore constitutes an indifference region in which no such guarantee is claimed.

\subsection{Estimating the DFE purity and state-distance
certification}
\label{sec: purity certification}

The DFE-overlap test developed above is a consistency test, but it does not, by itself, certify recovery of the ideal state. Exact recovery, $\widetilde{\rho}=\rho$, implies $\operatorname{Tr}(\widetilde{\rho}\rho)=1$, but the converse need not hold because the effective mitigated operator $\widetilde{\rho}$ is not necessarily positive semidefinite. For a pure ideal state $\rho$, the squared Hilbert--Schmidt distance is
\begin{equation}
\label{eq:hs_distance}
    D_{\mathrm{HS}}^{2}
    \equiv
    \left\|
        \widetilde{\rho}-\rho
    \right\|_{2}^{2}
    =
    \operatorname{Tr}(\widetilde{\rho}^{2})
    +1
    -2\operatorname{Tr}(\widetilde{\rho}\rho).
\end{equation}
Consequently, passing the overlap test establishes only compatibility with unit overlap. A complete estimate of $D_{\mathrm{HS}}^{2}$ additionally requires access to
\begin{equation}
    \mathcal{P}
    \equiv
    \operatorname{Tr}(\widetilde{\rho}^{2}).
\end{equation}
When $\widetilde{\rho}$ is not positive semidefinite, $\mathcal{P}$ is its squared Hilbert--Schmidt norm rather than the purity of a physical state.

Although $\widetilde{\rho}$ cannot, in general, be prepared as a physical state, its DFE purity can be estimated directly through the PEC quasiprobability decomposition. Let $I$ label a complete PEC trajectory, including all sampled Pauli corrections, and define
\begin{equation}
    \rho_I
    =
    \rho'_{\mathrm{noisy}}(\vec{P}'_I),
    \qquad
    \widetilde{p}_I
    =
    \frac{
        \left|
            \widetilde{a}(\vec{P}'_I)
        \right|
    }{
        \widetilde{Q}
    },
    \qquad
    s_I
    =
    \operatorname{sign}
    \!\left[
        \widetilde{a}(\vec{P}'_I)
    \right].
\end{equation}
The effective mitigated operator can then be written as
\begin{equation}
\label{eq:purity_pec_decomposition}
    \widetilde{\rho}
    =
    \widetilde{Q}
    \mathbb{E}_{I\sim\widetilde{p}}
    \!\left[
        s_I\rho_I
    \right].
\end{equation}
For two independently sampled trajectories $I,J\sim\widetilde{p}$, bilinearity of the Hilbert--Schmidt inner product gives
\begin{equation}
\label{eq:purity_pec_overlap}
    \mathcal{P}
    =
    \widetilde{Q}^{\,2}
    \mathbb{E}_{I,J\sim\widetilde{p}}
    \!\left[
        s_I s_J
        \operatorname{Tr}(\rho_I\rho_J)
    \right].
\end{equation}
Importantly, $\rho_I$ and $\rho_J$ are physical noisy output states even when their signed average $\widetilde{\rho}$ is not positive semidefinite.

The overlap in Eq.~\eqref{eq:purity_pec_overlap} can be measured using the global SWAP operator $S_{AB}$ acting between two $N$-qubit registers, since
\begin{equation}
    \operatorname{Tr}
    \!\left[
        S_{AB}(\rho_I\otimes\rho_J)
    \right]
    =
    \operatorname{Tr}(\rho_I\rho_J).
\end{equation}
In each experimental repetition, one independently samples two PEC trajectories, prepares the corresponding states on registers $A$ and $B$, and measures $S_{AB}$. Let $X\in\{-1,+1\}$ denote the resulting single-shot outcome, which satisfies
\begin{equation}
    \mathbb{E}[X\mid I,J]
    =
    \operatorname{Tr}(\rho_I\rho_J).
\end{equation}
Recording
\begin{equation}
    Z
    =
    \widetilde{Q}^{\,2}s_I s_J X
\end{equation}
gives the unbiased estimator
\begin{equation}
\label{eq:purity_pec_estimator}
    \widehat{\mathcal{P}}
    =
    \frac{1}{M}
    \sum_{m=1}^{M} Z_m,
    \qquad
    \mathbb{E}
    \!\left[
        \widehat{\mathcal{P}}
    \right]
    =
    \mathcal{P}.
\end{equation}
Unlike the DFE overlap, the DFE purity estimator does not require knowledge or classical simulation of the ideal target state.

Since $|Z|\leq\widetilde{Q}^{\,2}$, Hoeffding's inequality implies
\begin{equation}
\label{eq:purity_pec_concentration}
    \Pr\!\left[
        \left|
            \widehat{\mathcal{P}}-\mathcal{P}
        \right|
        \geq
        \epsilon_{\mathrm{pur}}
    \right]
    \leq
    2
    \exp\!\left[
        -
        \frac{
            M\epsilon_{\mathrm{pur}}^{2}
        }{
            2\widetilde{Q}^{\,4}
        }
    \right].
\end{equation}
It is therefore sufficient to choose
\begin{equation}
\label{eq:purity_pec_samples}
    M
    \geq
    \frac{
        2\widetilde{Q}^{\,4}
    }{
        \epsilon_{\mathrm{pur}}^{2}
    }
    \ln\!\left(
        \frac{2}{\delta_{\mathrm{pur}}}
    \right).
\end{equation}
to obtain additive precision $\epsilon_{\mathrm{pur}}$ with failure probability at most $\delta_{\mathrm{pur}}$. 

At fixed $\widetilde{Q}$, precision, and confidence, the number of paired measurements in Eq.~\eqref{eq:purity_pec_samples} has no explicit exponential dependence on the number of qubits. Each repetition instead requires two $N$-qubit registers and $\mathcal{O}(N)$ pairwise overlap operations. The overlap can be measured either using an ancilla-controlled SWAP test or through destructive Bell-basis measurements between corresponding qubit pairs ~\cite{Nakazato_2012,Tanaka_2013}.

The two-register implementation also introduces substantive experimental requirements. The trajectories must be sampled independently and, conditioned on their labels, the joint preparation must be well approximated by $\rho_I\otimes\rho_J$. Correlated noise or crosstalk between the registers can violate this assumption. Noise in the controlled-SWAP or Bell-measurement circuit can also introduce a systematic bias that is not controlled by Eq.~\eqref{eq:purity_pec_concentration}. These errors must therefore be calibrated, bounded, or mitigated separately.

If simultaneous access to two copies is unavailable, single-copy alternatives remain possible. The protocol of Ref.~\cite{Enk2012}, for example, can be adapted by replacing its required probabilities with mitigated estimators, but its measurement cost scales as $2^N$ and it requires global random measurements. Classical-shadow and local randomized-measurement approaches similarly avoid the second register but retain an exponential full-system cost~\cite{classical_shadows,Brydges_2019, error_mitigated_classical_shadows}. The two-copy PEC estimator therefore exchanges this additional exponential system-size dependence for doubled quantum memory and a joint overlap measurement.

Beyond its role in state-distance certification, the measured DFE purity also provides a directional diagnostic of consistent parameter miscalibration. The conditions under which deviations of $\widehat{Y}$ and $\widehat{\mathcal{P}}$ above or below one identify the direction of miscalibration are established in Sec.~\ref{sec: calibrating}.

\subsection{Reducing the PEC sampling overhead}
\label{sec: reducing PEC overhead} 
The DFE sampling cost for an effective mitigated state scales as $\widetilde{Q}^{\,2}$, while the DFE purity estimator of Sec.~\ref{sec: purity certification} scales as $\widetilde{Q}^{\,4}$. Reducing $\widetilde{Q}$ is therefore central to making both the overlap test and the purity estimation practical for deeper circuits. Separately, the measurement cost of standard DFE can be reduced by grouping commuting Pauli observables and sampling over groups rather than individual operators~\cite{BarberaRodriguez2025}; this does not, however, reduce the PEC normalization itself. Here we focus on two complementary strategies for reducing $\widetilde{Q}$.

The first approach can reduce $\widetilde{Q}$ at fixed circuit depth. Ref.~\cite{Scheiber2025reducedsampling} introduced error-propagated PEC (pPEC), which propagates the artificially inserted Pauli operators $P'$ to the beginning of the circuit and groups equivalent corrections. For Clifford circuits, this grouping was shown to reduce the sampling overhead. In Appendix~\ref{appendix ppec MG}, we extend the same construction to MG circuits under a sign-invariant noise assumption and prove
\begin{equation}
    \widetilde{Q}_{\mathrm{pPEC}}
    \leq
    \widetilde{Q}.
\end{equation}

The second strategy reduces $\widetilde{Q}$ by reducing the circuit depth. Dynamic circuits, which use mid-circuit measurements and feed-forward, can replace parts of a unitary circuit with measurements and classically conditioned operations, substantially reducing the depth of structured computations~\cite{Buhrman_2024,smith2024MPS,Yeo_2025,BaumerQFT}. For example, depth reductions for Matrix Product State preparation and transformation can be exponential in some settings ~\cite{smith2024MPS,Gunn_2025,Gunn_2026,Piroli_2021}. Since PEC also extends to dynamic circuits~\cite{gupta2023probabilisticerrorcancellationdynamic}, such depth reductions can directly reduce the associated sampling overhead.

The benefit of this approach is hardware dependent. Mid-circuit measurements and feed-forward can introduce additional crosstalk, spectator-qubit errors, and idling errors during classical-control latency~\cite{gupta2023probabilisticerrorcancellationdynamic}, with dynamical decoupling offering a possible mitigation of the latter ~\cite{shirgure2026errormitigationdynamiccircuits}. Dynamic-circuit rewritings are therefore advantageous only when the reduction in PEC overhead outweighs the additional errors they introduce; quantifying this trade-off is left for future work.

\section{Diagnosing the direction of noise-parameter miscalibration}
\label{sec: calibrating}

A central question for any error-mitigation protocol is how it responds when the assumed noise model deviates from the true one. In this section, we investigate when the DFE overlap can reveal not only the presence but also the direction or sign of such a miscalibration. In Sec.~\ref{sec: directionality general}, we define overestimation and underestimation for individual stochastic Pauli error sources and relate these definitions to several experimentally relevant noise models. We also distinguish these coefficientwise definitions from consistent miscalibration across several concatenated error sources. In Sec.~\ref{sec: ptm general}, we use the Pauli transfer matrix formalism to derive the corresponding directional responses. For the DFE overlap, the full response holds for Clifford circuits, while its extension to arbitrary circuits in the overestimation regime is perturbative. In contrast, the DFE purity exhibits the full directional response for arbitrary circuits and arbitrary parameter deviations under the stated effective-Pauli-fidelity conditions.

\subsection{Over- and underestimation of noise parameters}
\label{sec: directionality general}

For a single stochastic Pauli error source, we define the \emph{underestimation case} by
\begin{equation}
    \widetilde{c}_k(P)\leq c_k(P)
\end{equation}
for all layers $k$ and all nonidentity Pauli operators $P$, with at least one strict inequality. By normalization, this implies $\widetilde{c}_k(I)\geq c_k(I)$. Conversely, the \emph{overestimation case} is defined by
\begin{equation}
    \widetilde{c}_k(P)\geq c_k(P)
\end{equation}
for all $k$ and all $P\neq I$, again with at least one strict inequality, implying $\widetilde{c}_k(I)\leq c_k(I)$. These definitions exclude hybrid cases in which different Pauli coefficients of the same error source are simultaneously overestimated and underestimated.

Figure~\ref{fig:over_under_definition} illustrates these coefficientwise definitions for a single stochastic Pauli error source. Many experimentally relevant noise models can be described by one or a few physical error parameters whose miscalibration implies these coefficientwise conditions. This includes the global depolarizing, crosstalk, coherent over-rotation, dephasing, coherent phase-error, and amplitude-damping models considered in this work, either exactly or, for coherent-error models, within the corresponding perturbative regime.

In Appendix~\ref{appendix pauli coefficients crosstalk plus overotaion}, we derive the Pauli-channel coefficients and the corresponding Pauli-fidelity conditions for the crosstalk and coherent over-rotation models introduced in Sec.~\ref{sec: numerical protocol}. Here, $p_c$ denotes the crosstalk error probability defined in Eq.~\eqref{eq:crosstalk-channel}, while $\gamma$ denotes the over-rotation strength defined in Eq.~\eqref{eq:overrotation-generator}; $\widetilde p_c$ and $\widetilde\gamma$ denote the corresponding values assumed in the PEC procedure. For the crosstalk source, $\widetilde p_c<p_c$ corresponds to underestimation and $\widetilde p_c>p_c$ to overestimation. For the coherent over-rotation source, the analogous correspondence between $\widetilde\gamma$ and $\gamma$ holds within the perturbative small-error regime. The same correspondence also holds for dephasing, coherent phase errors, and amplitude damping after Pauli twirling~\cite{Piedrafita_2017}.

When several error sources are concatenated, we apply the above definitions separately to each constituent source. The coefficientwise ordering is generally not preserved by the concatenated Pauli channel. This, however, is not the quantity entering our analysis. Instead, Sec.~\ref{sec: ptm general} shows that the effective Pauli-fidelity inequalities defined in Eqs.~\eqref{eq: fidelities underestimation} and \eqref{eq: fidelities overestimation} are preserved under concatenation whenever all constituent error sources are miscalibrated in the same direction. Since these effective Pauli fidelities determine the asymmetric DFE-overlap and DFE-purity responses, we refer to such situations as \emph{consistent underestimation} and \emph{consistent overestimation}. Hybrid cases, in which some constituent error sources are overestimated while others are underestimated, are not covered by our analytical results, although we include such cases in the numerical simulations of Fig.~\ref{fig: protocol MG figure bigger than one}.

The analytical directionality results below apply in the
positive-fidelity regime, in which the true and estimated Pauli fidelities of every constituent error source are strictly positive for every layer $k$ and Pauli operator $P$. As shown in Appendix~\ref{appendix inverse channel}, a sufficient condition for this regime is that the identity coefficient of every true and estimated constituent Pauli channel is larger than $1/2$.

\subsection{Detecting miscalibration through the DFE overlap and DFE purity}
\label{sec: ptm general}

The Pauli transfer matrix (PTM) formalism provides a convenient description of the residual map $\mathcal{D}_{k,\mathrm{eff}}=\widetilde{\mathcal D}_k^{-1}\mathcal D_k$. Since both $\mathcal D_k$ and $\widetilde{\mathcal D}_k^{-1}$ are Pauli diagonal, so is $\mathcal D_{k,\mathrm{eff}}$~\cite{Roncallo_2023,Hantzko_2025}. As defined in Appendix~\ref{appendix inverse channel}, the diagonal PTM entries are the corresponding Pauli fidelities. It follows directly that the effective Pauli fidelities are given by
\begin{equation}
    f_{k,P,\mathrm{eff}}=\frac{f_{k,P}}{\widetilde f_{k,P}}.
\end{equation}
These determine the corresponding Pauli coefficients of the residual map through Eq.~(\ref{eq: coefficients from fidelities}) and are the central quantities in the analysis below. For a single error source, underestimation implies
\begin{equation}
\label{eq: fidelities underestimation}
0<f_{k,P,\mathrm{eff}}\le1
\qquad
\forall\,k,P,
\end{equation}
whereas overestimation implies
\begin{equation}
\label{eq: fidelities overestimation}
f_{k,P,\mathrm{eff}}\ge1
\qquad
\forall\,k,P,
\end{equation}
with at least one strict inequality in each case.

These effective-fidelity inequalities are preserved under consistent concatenation. Indeed, the PTM of a concatenation of Pauli channels is the product of the individual PTMs, so the effective Pauli fidelities multiply. Consequently, if every constituent error source is consistently underestimated, every factor satisfies $f_{k,P,\mathrm{eff}}\le1$, and so does their product. Likewise, consistent overestimation implies $f_{k,P,\mathrm{eff}}\ge1$ for the concatenated residual map. This preservation of the effective Pauli-fidelity inequalities---rather than the coefficientwise ordering of the concatenated Pauli channel---is the key property underlying the DFE-overlap and DFE-purity results derived below. No corresponding conclusion holds for hybrid concatenations.

Under consistent underestimation, including concatenated noise models for which all constituent error sources are consistently underestimated, the effective Pauli fidelities satisfy Eq.~\eqref{eq: fidelities underestimation}. Although this condition does not imply complete positivity of the residual map, it nevertheless leads to simple bounds on both the DFE purity and the DFE overlap. In Appendix~\ref{app:underestimation_overlap} we prove
\begin{equation}
\operatorname{Tr}(\widetilde{\rho}^{2})\le1.
\end{equation}
Since the ideal state is pure, $\operatorname{Tr}(\rho^2)=1$, the Cauchy--Schwarz inequality immediately yields
\begin{equation}
\operatorname{Tr}(\widetilde{\rho}\rho)\le1.
\end{equation}
Thus, consistent underestimation upper-bounds both the DFE purity and the DFE overlap by one for arbitrary circuits under stochastic Pauli noise.

The overestimation case is fundamentally different. Although consistent overestimation implies Eq.~\eqref{eq: fidelities overestimation}, this condition alone does not determine the sign of $\operatorname{Tr}(\widetilde{\rho}\rho)-1$. Consequently, unlike in the underestimation case, no general lower bound on the DFE overlap can be established for arbitrary circuits and arbitrary parameter deviations. Moreover, Appendix~\ref{appendix inverse channel} shows that any effective Pauli fidelity strictly larger than one necessarily implies that the corresponding residual map is not completely positive, since at least one effective Pauli coefficient becomes negative.

Nevertheless, two important positive results can be established. First, for Clifford circuits, Appendix~\ref{appendix: clifford proof} proves that
\begin{equation}
\operatorname{Tr}(\widetilde{\rho}\rho)\ge1
\end{equation}
under consistent overestimation. Together with the arbitrary-circuit underestimation bound above, this yields the full asymmetric DFE-overlap response for Clifford circuits and allows the protocol to diagnose the direction of the parameter miscalibration.

Second, Appendix~\ref{app:beyond-clifford} establishes a local perturbative extension of the overestimation response for arbitrary circuits. Introducing a parameter $\lambda$ that scales the deviation from perfect calibration, the effective Pauli fidelities admit the expansion
\begin{equation}
    f_{k,P,\mathrm{eff}}(\lambda)
    =
    1+\lambda r_{k,P}+o(\lambda),
\end{equation}
and the DFE overlap becomes
\begin{equation}
    \operatorname{Tr}\!\left[
        \widetilde{\rho}(\lambda)\rho
    \right]
    =
    1+\lambda A+o(\lambda),
\end{equation}
where
\begin{equation}
    A
    =
    \frac{1}{2^N}
    \sum_{k=1}^{s}
    \sum_{P\in\mathcal{P}_N}
    r_{k,P}\operatorname{Tr}(P\rho_k)^2.
    \label{eq: perturbative DFE coefficient}
\end{equation}
Consistent overestimation implies $r_{k,P}\ge0$ for every $k$ and $P$, and therefore $A\ge0$. Whenever $A>0$, the DFE overlap satisfies $\operatorname{Tr}(\widetilde{\rho}(\lambda)\rho)>1$ for all sufficiently small positive $\lambda$. If $A=0$, the first-order expansion is inconclusive and higher-order terms determine the response. Thus, the asymmetric overestimation response extends locally beyond Clifford circuits under the additional nonvanishing first-order response condition. 

The measurable DFE purity provides a complementary nonperturbative directional result for arbitrary circuits. Under consistent
underestimation, the argument above and Appendix~\ref{app:underestimation_overlap} establish
\begin{equation}
    \operatorname{Tr}(\widetilde{\rho}^{2})
    \leq
    1,
\end{equation}
whereas Appendix~\ref{appendix: purity proof} proves
\begin{equation}
    \operatorname{Tr}(\widetilde{\rho}^{2})
    \geq
    1
\end{equation}
under consistent overestimation. These inequalities hold for arbitrary circuits and arbitrary parameter deviations
whenever the corresponding effective Pauli fidelities
satisfy Eqs.~\eqref{eq: fidelities underestimation} and
\eqref{eq: fidelities overestimation}. Their validity under
consistent concatenation follows from the multiplicativity
of the effective Pauli fidelities.

Table~\ref{tab:over_under_summary_theoretical} summarizes
the analytical results established in this section, while
Table~\ref{tab:over_under_summary_experimental} summarizes
their finite-sample diagnostic interpretation.

\section{Applications}
In this section, we illustrate the testing and diagnosis protocol in two complementary settings. Sec.~\ref{sec: numerical protocol} applies the protocol numerically to Matchgate circuits subject to a concatenation of crosstalk and coherent over-rotation errors, demonstrating both the DFE overlap-based consistency test and the miscalibration diagnostic for consistent miscalibration. Sec.~\ref{sec: global depolarsing} then considers the global depolarizing model, where the effective mitigated state and the resulting DFE overlap can be computed exactly, providing a fully analytical illustration of the directional response derived in Sec.~\ref{sec: calibrating}.

\label{sec: applications}

\begin{figure}[t]
    \centering
    \includegraphics[width=\linewidth]{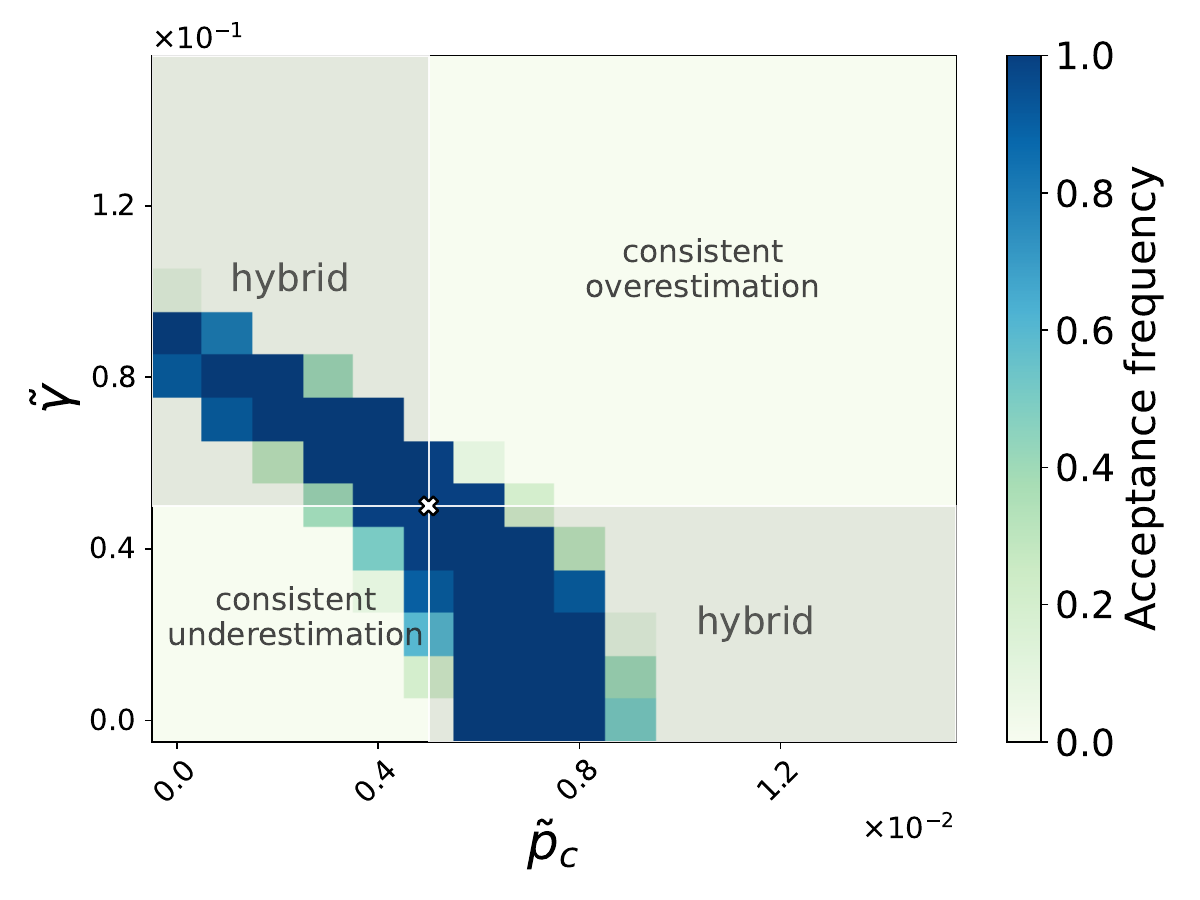}
    \caption{Acceptance frequency of the DFE protocol as a function of the deviated error parameters $\widetilde{p}_c$ and $\widetilde{\gamma}$. The white cross marks the true error parameters, $p_c = 0.005$ and $\gamma = 0.05$. We consider a random $N = 12$ qubit circuit $UU^\dagger$, where $U$ is an eight-gate MG circuit. The DFE overlap estimator $\widehat{Y}$ is obtained by setting $\epsilon = 0.05$ and $\delta = 0.1$. The acceptance frequency is evaluated by repeating the DFE protocol ten times for the same circuit and applying the acceptance criterion given in Eq.~(\ref{eq: protocol}).}
    \label{fig: protocol MG figure}
\end{figure}

\subsection{The testing and diagnosis protocol under crosstalk plus over-rotation error channels}
\label{sec: numerical protocol}
In this section, we illustrate the proposed protocol using a physically relevant error model consisting of a concatenation of coherent over-rotation and crosstalk errors. The over-rotation is parameterized by $\gamma$, while the crosstalk error is parameterized by $p_c$. We introduce controlled deviations of these parameters to test the consistency of the assumed error model and the resulting PEC procedure. The same combination of crosstalk and over-rotation errors was considered in Ref.~\cite{carrasco2024} in the context of a related protocol. We use a gate-dependent over-rotation model to make the test more representative, while the crosstalk model is motivated by noise characteristics observed in ion-trap quantum processors~\cite{Heu_en_2023_ion_trap}.

We consider MG circuits composed of two-qubit MGs of the form
\begin{equation}
    G_k = M(\beta_k,\tau_k)
    = e^{i\beta_k X_iX_j}e^{i\tau_k Y_iY_j},
\end{equation}
which represent general MGs up to local phase gates~\cite{carrasco2024}. After each MG, we apply an error channel consisting of the concatenation of crosstalk and over-rotation noise. Both error channels are converted into Pauli-stochastic channels of the form in Eq.~(\ref{eq: stochastic pauli channel}) by Pauli twirling~\cite{van_den_Berg_2023,Wallman_2016}.

All numerical experiments in this section use the standard gate-local PEC construction and the estimator defined in Eq.~\eqref{eq: OPEC}. In particular, we do not apply the pPEC grouping or the dynamic-circuit reductions discussed in Sec.~\ref{sec: reducing PEC overhead}. 

We model crosstalk between neighboring qubits by
\begin{align}
    \mathcal{C}_{ij}(\rho)
    &= (1-p_c)\rho
    +\frac{p_c}{4}\Big(
    X_iX_j\rho X_iX_j
    +X_iY_j\rho X_iY_j \notag\\
    &\qquad\qquad
    +Y_iX_j\rho Y_iX_j
    +Y_iY_j\rho Y_iY_j
    \Big).
    \label{eq:crosstalk-channel}
\end{align}
The crosstalk channel acts between each targeted qubit and each of its neighboring qubits whenever a two-qubit gate is applied. Thus, for a gate acting on qubits $i$ and $i+1$, the total crosstalk channel is
\begin{equation}
    \mathcal{C}_{i,i-1}\circ
    \mathcal{C}_{i,i+1}\circ
    \mathcal{C}_{i+1,i}\circ
    \mathcal{C}_{i+1,i+2}.
\end{equation}
We set the true crosstalk parameter to $p_c=0.005$, following Ref.~\cite{carrasco2024}.

For each gate $G_k=e^{iH_k}$, with
\begin{equation}
    H_k=\sum_P\beta_{k,P}P,
\end{equation}
we introduce a coherent over-rotation generated by
\begin{equation}
    H_k^{(\mathrm{OR})}
    =\gamma\sum_P|\beta_{k,P}|P.
    \label{eq:overrotation-generator}
\end{equation}
The implemented gate is therefore
\begin{equation}
    O_kG_k
    =e^{iH_k^{(\mathrm{OR})}}e^{iH_k}.
\end{equation}
After Pauli twirling, this gives a stochastic Pauli channel whose coefficients depend on the individual gate $k$, while the same global parameter $\gamma$ is used for all gates. Following Ref.~\cite{carrasco2024}, we set the true over-rotation parameter to $\gamma=0.05$.

We first consider the acceptance test in Fig.~\ref{fig: protocol MG figure}. We apply the protocol to the circuit $UU^\dagger$, whose ideal output is a well-conditioned state with $\alpha=1$. The circuit $U$ contains eight MGs with the angles $\beta_k$ and $\tau_k$ and the pairs of neighboring qubits on which the gates act chosen randomly. The circuit acts on $N=12$ qubits. We set $\epsilon=0.05$ and $\delta=0.1$ and, for each pair of assumed parameters $(\widetilde p_c,\widetilde\gamma)$, repeat the DFE protocol ten times. The plotted acceptance frequency is the fraction of these repetitions that pass the overlap test in Eq.~\eqref{eq: protocol}. To verify that this behavior is not specific to the circuit instance shown in Fig.~\ref{fig: protocol MG figure}, Appendix~\ref{appendix:extended_MG_numerics} repeats the analysis for ten independently generated MG circuits with different computational-basis product input states. The extended results show the same qualitative acceptance structure across the tested ensemble. 

For the correctly calibrated parameters, $\widetilde p_c=p_c$ and $\widetilde\gamma=\gamma$, all ten repetitions pass the overlap test. For sufficiently large deviations from the true parameters, no repetition passes. Between these regimes, we observe a semi-circular acceptance region in which an overestimation of one parameter can partially compensate an underestimation of the other, keeping the DFE overlap close to one. A similar structure was observed in Ref.~\cite{carrasco2024}. This hybrid-miscalibration region illustrates that an overlap close to one does not by itself imply that the individual error parameters are correctly calibrated.

Figure~\ref{fig: protocol MG figure bigger than one} examines the direction of the DFE-overlap response for the same circuit and error model. In the parameter range considered, we find $\operatorname{Tr}(\widetilde{\rho}\rho)>1$ when both error sources are overestimated and $\operatorname{Tr}(\widetilde{\rho}\rho)<1$ when both are underestimated. The latter behavior follows from the general underestimation result, while the former is consistent with the local perturbative result of Appendix~\ref{app:beyond-clifford}, which predicts $\operatorname{Tr}(\widetilde{\rho}\rho)>1$ for sufficiently small consistent overestimation whenever the first-order response coefficient $A$ is strictly positive. Thus, the figure provides a numerical illustration of the directional response in a non-Clifford circuit rather than establishing a general overestimation bound. As in Fig.~\ref{fig: protocol MG figure}, we set $\epsilon=0.05$ and $\delta=0.1$. Values within the statistical interval $1\pm2\epsilon$ are treated as inconclusive, and each displayed value is averaged over ten independent repetitions of the DFE protocol. Appendix~\ref{appendix:extended_MG_numerics} reports the corresponding directional analysis for the same ensemble of ten independently generated MG circuits and computational-basis product input states. Across the ensemble, we observe the same qualitative asymmetric response: the DFE overlap lies below one on the consistent-underestimation side and above one on the consistent-overestimation side. 

\begin{figure}[b]
    \centering
    \includegraphics[width=\linewidth]{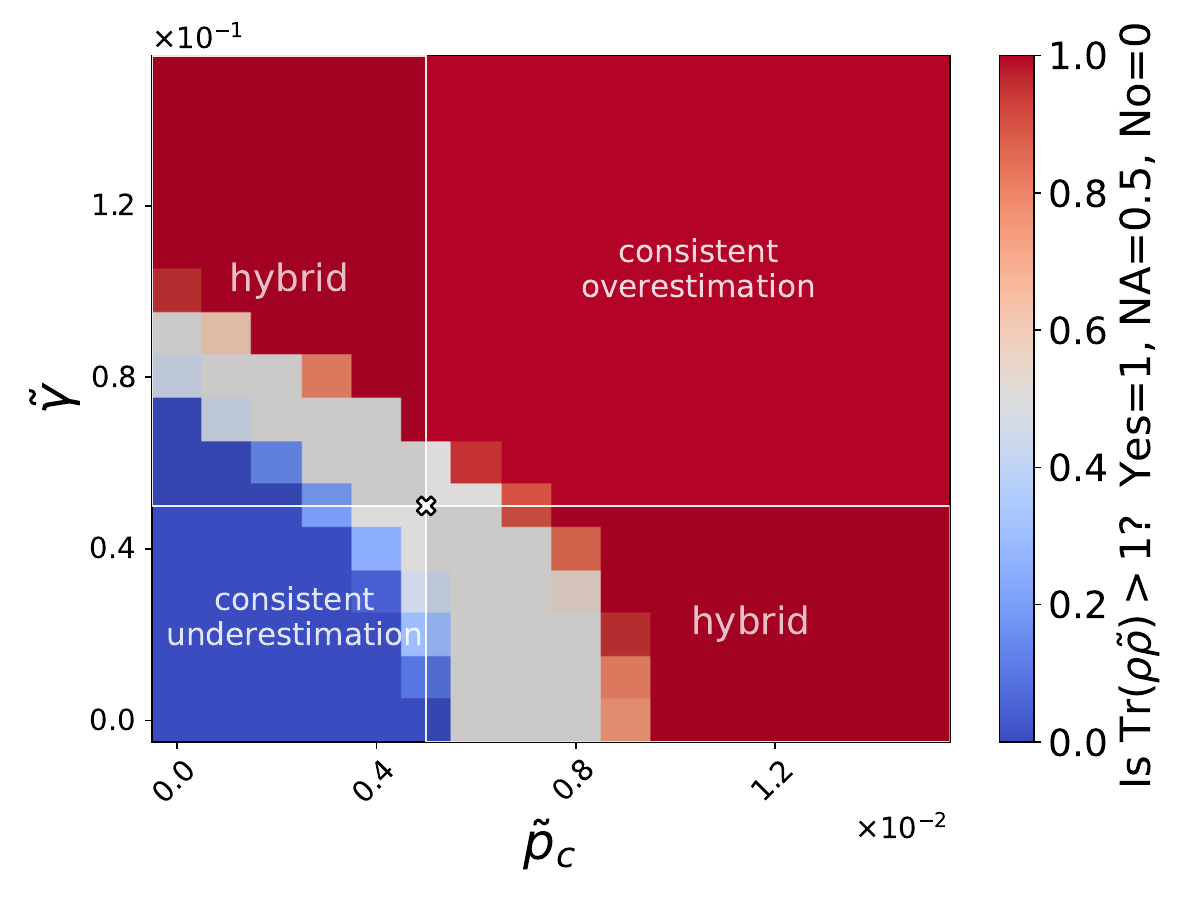}
    \caption{Variable indicating whether the measured overlap is larger (1) or smaller (0) than one within the given precision ($2\epsilon$ as indicated by Eq.~(\ref{main bound precision Y})) as a function of the deviated error parameters $\widetilde{p}_c$ and $\widetilde{\gamma}$. The grey region (0.5) corresponds to the values where one cannot say if the DFE overlap is larger or smaller than one, i.e., the measured overlap is in the range $1\pm2\epsilon$. The white cross marks the true error parameters, $p_c = 0.005$ and $\gamma = 0.05$. We consider a random $N = 12$ qubit circuit $UU^\dagger$, where $U$ is an eight-gate MG circuit. The DFE overlap estimator $\widehat{Y}$ is obtained using $\epsilon = 0.05$ and $\delta = 0.1$. The values shown are obtained by averaging over ten repetitions of the DFE protocol.}
    \label{fig: protocol MG figure bigger than one}
\end{figure}

\subsection{Parameter miscalibration direction in the global depolarizing error model}
\label{sec: global depolarsing}

The global depolarizing noise channel acting after each noisy gate is described by
\begin{equation}
\mathcal{D}(\mathcal{G}_k)(\cdot)
= (1 - p_e)\, G_k(\cdot) G_k^\dagger
+ \frac{p_e}{4^N} \sum_{P \in \mathcal{P}_N} P(\cdot) P,
\end{equation}
where $p_e$ fully characterizes the error model and is independent of the gate index $k$. In Appendix~\ref{appendix global depolarizing mitigated state}, we show that, when PEC is implemented using a deviated parameter $\widetilde{p}_e$, the effective mitigated state is
\begin{equation}
    \widetilde{\rho} = (1 - p_{\mathrm{PEC}})^s \rho
    + \big[1 - (1 - p_{\mathrm{PEC}})^s \big]
    \frac{I}{2^N},
\end{equation}
where $s$ is the number of gates and
\begin{equation}
    p_{\mathrm{PEC}} = \frac{p_e - \widetilde{p}_e}{1 - \widetilde{p}_e}.
\end{equation}
For $\widetilde{p}_e < p_e$, we have $p_{\mathrm{PEC}}>0$, so $\widetilde{\rho}$ remains a valid density matrix and has the form of a globally depolarized version of $\rho$. For $\widetilde{p}_e > p_e$, we have $p_{\mathrm{PEC}}<0$, and the coefficient of $I/2^N$ becomes negative, so $\widetilde{\rho}$ is generally not a valid quantum state.

The DFE overlap between $\rho$ and $\widetilde{\rho}$ is
\begin{equation}
    \label{eq: fidelity global depo}
    \mathrm{Tr}(\widetilde{\rho}\rho)
    = (1 - p_{\mathrm{PEC}})^s
    + \big[1 - (1 - p_{\mathrm{PEC}})^s\big]
    \frac{1}{2^N},
\end{equation}
which is independent of the particular pure ideal state $\rho$ and hence of the ideal unitary being implemented. For correctly calibrated parameters, $p_e=\widetilde{p}_e$ and therefore $p_{\mathrm{PEC}}=0$, leading to $\mathrm{Tr}(\widetilde{\rho}\rho)=1$.

Equation~\eqref{eq: fidelity global depo} also shows that the overlap reveals the direction of the parameter miscalibration. Underestimation, $\widetilde{p}_e<p_e$, gives $p_{\mathrm{PEC}}>0$ and hence $\mathrm{Tr}(\widetilde{\rho}\rho)<1$, whereas overestimation, $\widetilde{p}_e>p_e$, gives $p_{\mathrm{PEC}}<0$ and hence $\mathrm{Tr}(\widetilde{\rho}\rho)>1$. Thus, in the global depolarizing model, the sign of the deviation of the DFE overlap from one directly indicates whether the error parameter was underestimated or overestimated. Figure~\ref{fig: fidelity global} illustrates this behavior by showing the exact overlap given in Eq.~\eqref{eq: fidelity global depo} as a function of the assumed parameter $\widetilde{p}_e$ for different numbers of gates $s$. We use $N=50$ qubits and a true error parameter $p_e=0.02$.

\begin{figure}
    \centering
    \includegraphics[width=\linewidth]{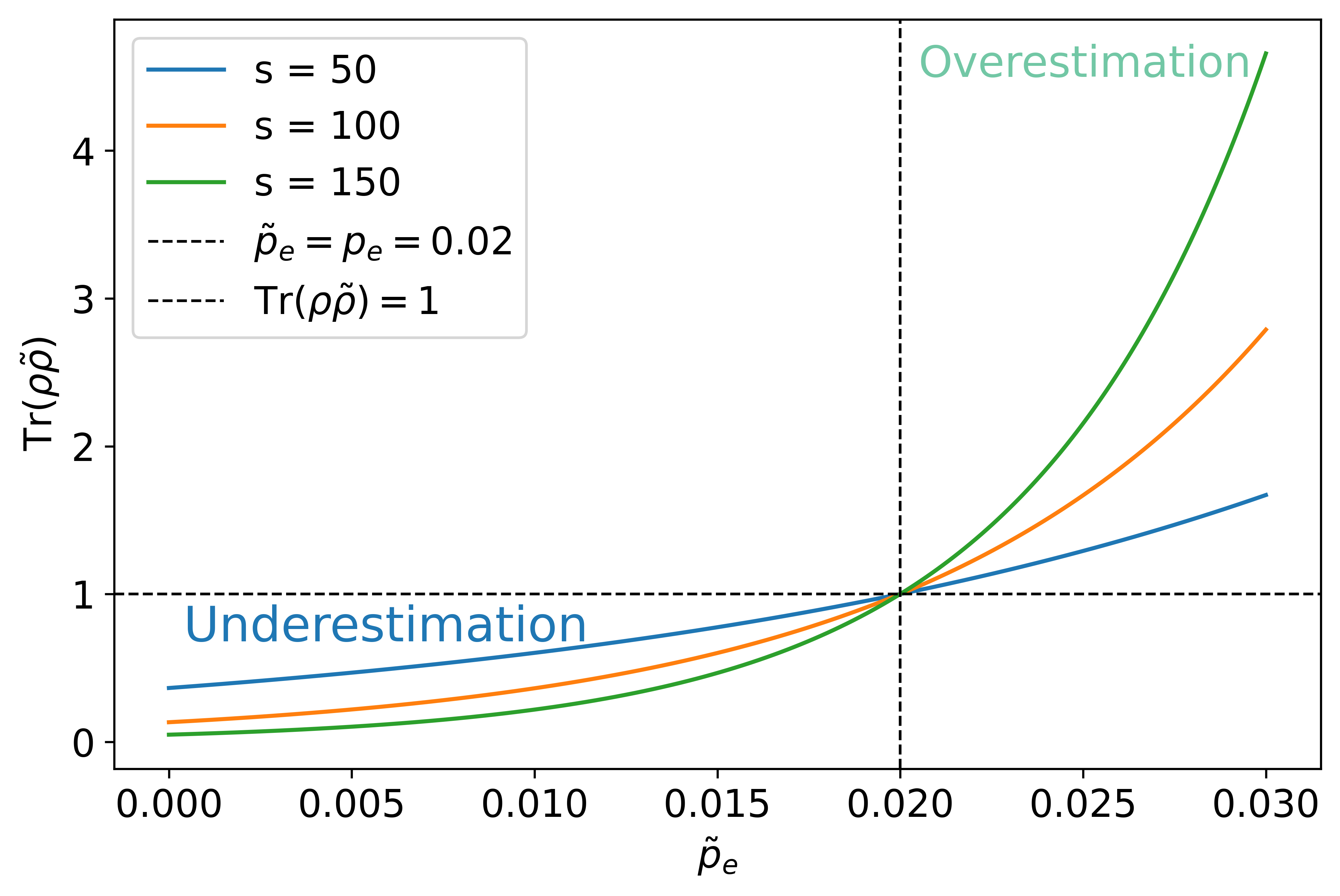}
    \caption{Exact DFE overlap between the ideal state and the effective mitigated state for the global depolarizing model as a function of the deviated error parameter $\widetilde{p}_e$. We set the true error parameter to $p_e = 0.02$ and plot the DFE overlap for $N = 50$ qubits and $s = 50$, $100$, and $150$ gates.}
    \label{fig: fidelity global}
\end{figure}

Equation~\eqref{eq: fidelity global depo} also provides a quantitative recalibration formula. Substituting the definition of \(p_{\mathrm{PEC}}\) gives
\begin{equation}
    \operatorname{Tr}(\widetilde{\rho}\rho)
    =
    2^{-N}
    +
    \left(1-2^{-N}\right)
    \left(\frac{1-p_e}{1-\widetilde{p}_e}\right)^{s}.
\end{equation}
On the physical branch \(0\leq p_e,\widetilde{p}_e<1\), this relation can be inverted as
\begin{equation}
    p_e
    =
    1-
    \left(1-\widetilde{p}_e\right)
    \left(
        \frac{\operatorname{Tr}(\widetilde{\rho}\rho)-2^{-N}}
        {1-2^{-N}}
    \right)^{1/s}.
    \label{eq:global-depolarizing-recalibration}
\end{equation}
Thus, within the global depolarizing model, the overlap determines not only the direction of the miscalibration but also the true error parameter. For finite data, replacing the exact overlap by \(\widehat{Y}\) yields a corresponding plug-in estimate, with uncertainty obtained by propagating the confidence interval in Eq.~\eqref{main bound precision Y}.

\section{Conclusions \& Outlook}
\label{sec: conclusions}

We have introduced a protocol for testing the consistency of the error parameters used in probabilistic error cancellation. The protocol combines probabilistic error cancellation with direct fidelity estimation applied to the effective mitigated state and can be implemented using classically simulable surrogate circuits that preserve key features of the target computation. In this way, it provides an end-to-end test of a PEC implementation directly on hardware.

Our main analytical result is that the DFE overlap can serve not only as a witness of parameter miscalibration but, under suitable conditions, also as a diagnostic of its direction. For global depolarizing noise and for Clifford circuits subject to stochastic Pauli noise, consistent overestimation and underestimation lead to overlaps above and below one, respectively. For arbitrary circuits, the underestimation result remains general, while the corresponding overestimation result is established perturbatively for sufficiently small deviations.

We further showed that the purity of the effective mitigated operator provides a stronger and more general diagnostic. Under consistent miscalibration, its deviation from unity identifies the direction of the miscalibration for arbitrary circuits and arbitrary parameter deviations. Together, the overlap and purity determine the Hilbert--Schmidt distance between the effective mitigated operator and the ideal state, thereby extending the consistency test into a quantitative state-distance certificate.

More broadly, our results provide experimentally accessible tools for validating the assumptions underlying noise-aware quantum error mitigation and for diagnosing inaccuracies in the error parameters used by PEC. For future research several open questions remain; for instance, it would be interesting to study in more detail the role of mid-circuit measurements in probabilistic error cancellation, and the resulting trade-off between their additional errors and reduction of PEC overhead. It would be furthermore interesting to combine PEC with Hamiltonian and Lindbladian learning techniques~\cite{Olsacher2025,Kraft2026,vandenBerg2026}. Moreover, it would be interesting to study the effect of stochastic error parameters that include statistical uncertainty about the error parameters used for PEC.

\begin{acknowledgments}
We acknowledge helpful discussions with Marc Langer. We acknowledge the BMW endowment fund. This publication has received funding from the European Union's Horizon 2020 HORIZON Research and Innovation Actions Programme under the calls HORIZON-CL4-2022-QUANTUM-02-SGA via Grant Agreement No. 101113690 (PASQuanS2.1) and HORIZON-CL4-2021-DIGITAL-EMERGING-02-10 via Grant Agreement No. 101080085 (QCFD). The codes to generate the numerical simulations are available from the corresponding author upon reasonable request.

{\bf AI statement:} We have used ChatGPT 5.6 and Claude Opus 5 for text revisions, grammar checking, consistency review, and editorial suggestions. ChatGPT 5.6 was additionally used as an interactive assistant to test derivations, and explore alternative formulations. In particular, it suggested the quasiprobability two-copy estimator for the purity of the effective PEC operator. All resulting arguments and statements were independently checked, and revised where necessary. The authors retained full responsibility for all scientific content and final manuscript text.
\end{acknowledgments}

\appendix
\section{A sufficient condition for the existence of the inverse Pauli map}
\label{appendix inverse channel}
In this appendix, we provide a sufficient condition for the invertibility of the Pauli error channel
\begin{equation}
    \mathcal{D}_k(\cdot)
    =
    \sum_{P\in\mathcal{P}_n}
    c_k(P)\,P(\cdot)P.
\end{equation}
Since conjugation by a Pauli operator maps every Pauli operator $P$ to either $P$ or $-P$, the Pauli operators are eigenoperators of $\mathcal{D}_k$:
\begin{equation}
    \mathcal{D}_k(P)
    =
    f_{k,P}P,
\end{equation}
where $f_{k,P}$ is the Pauli fidelity associated with $P$. It follows immediately that $\mathcal{D}_k$ is invertible if and only if $f_{k,P}\neq0$ for every $P\in\mathcal{P}_n$, in which case
$\mathcal{D}_k^{-1}(P)=f_{k,P}^{-1}P$. Since the Pauli-conjugation maps $P(\cdot)P$ form a basis for Pauli-diagonal linear maps, the inverse is itself Pauli diagonal and can be written as
\begin{equation}
    \mathcal{D}_k^{-1}(\cdot)
    =
    \sum_{P\in\mathcal{P}_n}
    a_k(P)\,P(\cdot)P.
\end{equation}
The coefficients $a_k(P)$ are not necessarily nonnegative, so $\mathcal{D}_k^{-1}$ is not, in general, a completely positive map and therefore does not necessarily represent a physical quantum channel.

To express this invertibility condition and the inverse coefficients $a_k(P)$ directly in terms of the channel coefficients $c_k(P)$, we employ the Pauli transfer matrix (PTM) formalism, which will also be used in later derivations~\cite{Roncallo_2023,Hantzko_2025}. The PTM $T^{(k)}$ of $\mathcal{D}_k$ is diagonal in the Pauli basis, with matrix elements
\begin{align}
    T^{(k)}_{ij}
    &=
    \frac{1}{2^n}
    \operatorname{Tr}
    \!\left[
        P_i\mathcal{D}_k(P_j)
    \right] \nonumber\\
    &=
    \delta_{ij}
    \sum_{P\in\mathcal{P}_n}
    c_k(P)(-1)^{J_{P,P_j}},
\end{align}
where $P_i$ and $P_j$ denote elements of the phase-free Pauli group $\mathcal{P}_n=\{I,X,Y,Z\}^{\otimes n}$.

For Pauli strings $P_a=\bigotimes_{r=1}^n P_{a,r}$ and $P_b=\bigotimes_{r=1}^n P_{b,r}$, we define
\begin{equation}
    \label{eq: appendix J}
    J_{P_a,P_b}
    =
    \left|
    \left\{
        r:
        P_{a,r}P_{b,r}
        =
        -P_{b,r}P_{a,r}
    \right\}
    \right|.
\end{equation}
Thus, $J_{P_a,P_b}$ counts the number of qubit positions at which the corresponding single-qubit Pauli operators anticommute. In particular,
\begin{equation}
    P_aP_b
    =
    (-1)^{J_{P_a,P_b}}P_bP_a,
\end{equation}
so only the parity of $J_{P_a,P_b}$ determines whether the two Pauli strings commute or anticommute.

The diagonal elements of $T^{(k)}$ are the Pauli fidelities,
\begin{equation}
    f_{k,P}=T^{(k)}_{P,P}.
\end{equation}
They are related to the Pauli coefficients by
\begin{equation}
    \label{eq: fidelity definition appendix}
    f_{k,P}
    =
    \sum_{P'\in\mathcal{P}_n}
    c_k(P')(-1)^{J_{P,P'}},
\end{equation}
while the inverse relation is~\cite{van_den_Berg_2023}
\begin{equation}
    \label{eq: coefficients from fidelities}
    c_k(P)
    =
    \frac{1}{4^n}
    \sum_{P'\in\mathcal{P}_n}
    (-1)^{J_{P,P'}}
    f_{k,P'}.
\end{equation}

Since $T^{(k)}$ is diagonal, it is invertible if and only if none of its diagonal entries vanishes. When this condition is satisfied,
\begin{equation}
    \left[(T^{(k)})^{-1}\right]_{ij}
    =
    \delta_{ij}\frac{1}{f_{k,P_i}}.
\end{equation}
Applying the inverse transformation in Eq.~(\ref{eq: coefficients from fidelities}) therefore gives
\begin{equation}
    a_k(P)
    =
    \frac{1}{4^n}
    \sum_{P'\in\mathcal{P}_n}
    (-1)^{J_{P,P'}}
    \frac{1}{f_{k,P'}}.
\end{equation}
Hence, the inverse map exists if and only if $f_{k,P}\neq0$ for every $P\in\mathcal{P}_n$.

The same fidelity--coefficient relation also provides a simple criterion for detecting when a residual Pauli-diagonal map is not completely positive. Consider
\begin{equation}
    \mathcal{D}_{k,\mathrm{eff}}
    =
    \widetilde{\mathcal{D}}_k^{-1}\mathcal{D}_k,
\end{equation}
with Pauli coefficients $c_{k,\mathrm{eff}}(R)$ and Pauli fidelities $f_{k,P,\mathrm{eff}}$. This map is trace preserving and therefore
\begin{equation}
    \sum_{R\in\mathcal{P}_n}
    c_{k,\mathrm{eff}}(R)
    =
    1.
\end{equation}
If all its Pauli coefficients were nonnegative, then Eq.~\eqref{eq: fidelity definition appendix} would give, for every $P$,
\begin{align}
    f_{k,P,\mathrm{eff}}
    &=
    \sum_{R\in\mathcal{P}_n}
    c_{k,\mathrm{eff}}(R)
    (-1)^{J_{P,R}}
    \notag\\
    &\leq
    \sum_{R\in\mathcal{P}_n}
    c_{k,\mathrm{eff}}(R)
    =
    1.
\end{align}
Consequently, if $f_{k,P,\mathrm{eff}}>1$ for at least one Pauli operator $P$, then at least one coefficient $c_{k,\mathrm{eff}}(R)$ must be negative. Since complete positivity of a Pauli-diagonal map requires all its Pauli coefficients to be nonnegative, the residual map is then not completely positive.

We now show that $c_k(I^{\otimes n})>1/2$ is sufficient to satisfy this condition. From Eq.~(\ref{eq: fidelity definition appendix}),
\begin{equation}
    f_{k,P}
    =
    c_k(I^{\otimes n})
    +
    \sum_{P'\in
    \mathcal{P}_n\setminus\{I^{\otimes n}\}}
    c_k(P')(-1)^{J_{P,P'}}.
\end{equation}
For every $P$, this quantity obeys
\begin{align}
    f_{k,P}
    &\geq
    c_k(I^{\otimes n})
    -
    \sum_{P'\in
    \mathcal{P}_n\setminus\{I^{\otimes n}\}}
    c_k(P') \nonumber\\
    &=
    2c_k(I^{\otimes n})-1,
\end{align}
where we used the normalization condition $\sum_{P'\in\mathcal{P}_n}c_k(P')=1$. Consequently, if
\begin{equation}
    c_k(I^{\otimes n})>\frac{1}{2},
\end{equation}
then $f_{k,P}>0$ for every $P\in\mathcal{P}_n$, and the inverse map $\mathcal{D}_k^{-1}$ exists.

We emphasize that $c_k(I^{\otimes n})>1/2$ is a sufficient but not necessary condition. A Pauli channel with $c_k(I^{\otimes n})\leq1/2$ may still be invertible, provided that none of its Pauli fidelities vanishes.

Finally, although the inverse is algebraically well defined under the conditions above, its Pauli representation contains $4^n$ coefficients and therefore grows exponentially with the number $n$ of qubits on which the error channel acts. Computing and storing the inverse is consequently practical only when the noise is assumed to be local, so that each error channel acts on a bounded number of qubits.

\section{Measurement budget for DFE applied to the
effective mitigated state}
\label{appendix DFE bounds}

In this appendix, we derive the measurement budget required
to estimate the DFE overlap of the effective mitigated
operator. We first consider a general pure ideal state and
then specialize the analysis to an
$\alpha$-well-conditioned ideal state, for which the
exponential dependence on the number of qubits can be
avoided.

The finite number $m_P$ of measurement samples used to
estimate each $\widetilde{P}_{\mathrm{PEC}}$, together with
the finite number $\ell$ of sampled Pauli operators,
introduces two distinct sources of statistical uncertainty
into the estimator
\begin{equation}
    \widehat{Y}
    =
    \frac{1}{\ell}
    \sum_{i=1}^{\ell}
    \widehat{X}_{P_i}.
\end{equation}
The Pauli operators $P_i$ are sampled independently
according to the relevance distribution
$\chi_\rho(P)^2$. We denote the support of this distribution
by
\begin{equation}
    \mathcal{S}_\rho
    =
    \left\{
        P\in\mathcal{P}_N:
        \chi_\rho(P)\neq0
    \right\}.
\end{equation}
All sums and sampled Pauli operators appearing below are
implicitly restricted to $\mathcal{S}_\rho$ whenever a
division by $\chi_\rho(P)$ or
$\operatorname{Tr}[P\rho]$ occurs. We define
\begin{equation}
    X_P
    =
    \frac{
        \operatorname{Tr}[\widetilde{\rho}P]
    }{
        \operatorname{Tr}[\rho P]
    }
\end{equation}
and denote by
\begin{equation}
\label{eq: appendix Y no error in X}
    \overline{Y}
    =
    \frac{1}{\ell}
    \sum_{i=1}^{\ell}
    X_{P_i}
\end{equation}
the estimator obtained when each $X_P$ is evaluated
exactly. Thus, the difference between $\overline{Y}$ and
$\operatorname{Tr}(\widetilde{\rho}\rho)$ describes the
uncertainty arising from the finite number $\ell$ of
sampled Pauli operators, whereas the difference between
$\widehat{Y}$ and $\overline{Y}$ describes the uncertainty
arising from the finite number of measurements used to
estimate each $X_P$.

\subsection{General pure ideal states}
\label{appendix DFE bounds general}

We first bound the uncertainty arising from the finite
number $\ell$ of sampled Pauli operators. Because $\rho$ is
pure and
$\chi_\rho(P)=\operatorname{Tr}[P\rho]/\sqrt{2^N}$, the
second moment of $X_P$ satisfies~\cite{Kliesch_2021}
\begin{align}
    \mathbb{E}_{P\sim\chi_\rho^2}[X_P^2]
    &=
    \sum_{P\in\mathcal{S}_\rho}
    \chi_\rho(P)^2
    \frac{
        \operatorname{Tr}[\widetilde{\rho}P]^2
    }{
        \operatorname{Tr}[\rho P]^2
    }
    \nonumber\\
    &=
    \frac{1}{2^N}
    \sum_{P\in\mathcal{S}_\rho}
    \operatorname{Tr}[\widetilde{\rho}P]^2
    \nonumber\\
    &\leq
    \frac{1}{2^N}
    \sum_{P\in\mathcal{P}_N}
    \operatorname{Tr}[\widetilde{\rho}P]^2
    \nonumber\\
    &=
    \operatorname{Tr}[\widetilde{\rho}^{\,2}],
\end{align}
where the last equality follows from the orthogonality of
the Pauli basis. Consequently,
\begin{equation}
    \operatorname{Var}_{P\sim\chi_\rho^2}[X_P]
    \leq
    \mathbb{E}_{P\sim\chi_\rho^2}[X_P^2]
    \leq
    \operatorname{Tr}[\widetilde{\rho}^{\,2}].
\end{equation}

Unlike a physical density operator, the effective
mitigated operator need not have purity upper-bounded by
one. Nevertheless, its PEC decomposition gives
\begin{equation}
    \widetilde{\rho}
    =
    \sum_{\vec{P}'}
    \widetilde{a}(\vec{P}')
    \rho'_{\mathrm{noisy}}(\vec{P}').
\end{equation}
For any two physical density operators,
\begin{equation}
    0
    \leq
    \operatorname{Tr}
    \!\left[
        \rho'_{\mathrm{noisy}}(\vec{P}')
        \rho'_{\mathrm{noisy}}(\vec{P}'')
    \right]
    \leq
    1.
\end{equation}
It follows that
\begin{align}
    \operatorname{Tr}[\widetilde{\rho}^{\,2}]
    &=
    \sum_{\vec{P}'}
    \sum_{\vec{P}''}
    \widetilde{a}(\vec{P}')
    \widetilde{a}(\vec{P}'')
    \operatorname{Tr}
    \!\left[
        \rho'_{\mathrm{noisy}}(\vec{P}')
        \rho'_{\mathrm{noisy}}(\vec{P}'')
    \right]
    \nonumber\\
    &\leq
    \sum_{\vec{P}'}
    \sum_{\vec{P}''}
    \left|
        \widetilde{a}(\vec{P}')
        \widetilde{a}(\vec{P}'')
    \right|
    \nonumber\\
    &=
    \left(
        \sum_{\vec{P}'}
        \left|
            \widetilde{a}(\vec{P}')
        \right|
    \right)^2
    =
    \widetilde{Q}^{\,2}.
\end{align}

Since the Pauli operators $P_i$ are sampled independently,
\begin{equation}
    \operatorname{Var}[\overline{Y}]
    =
    \frac{
        \operatorname{Var}_{P\sim\chi_\rho^2}[X_P]
    }{
        \ell
    }
    \leq
    \frac{
        \widetilde{Q}^{\,2}
    }{
        \ell
    }.
\end{equation}
Chebyshev's inequality therefore gives
\begin{align}
    \Pr\!\left[
        \left|
            \overline{Y}
            -
            \operatorname{Tr}(\widetilde{\rho}\rho)
        \right|
        \geq
        \epsilon
    \right]
    &\leq
    \frac{
        \operatorname{Var}[\overline{Y}]
    }{
        \epsilon^2
    }
    \nonumber\\
    &\leq
    \frac{
        \widetilde{Q}^{\,2}
    }{
        \ell\epsilon^2
    }.
\end{align}
Hence, choosing
\begin{equation}
\label{eq:general_ell_bound}
    \ell
    \geq
    \frac{
        \widetilde{Q}^{\,2}
    }{
        \epsilon^2\delta
    }
\end{equation}
ensures
\begin{equation}
    \Pr\!\left[
        \left|
            \overline{Y}
            -
            \operatorname{Tr}(\widetilde{\rho}\rho)
        \right|
        \geq
        \epsilon
    \right]
    \leq
    \delta.
\end{equation}

We next bound the uncertainty arising from the finite
number $m_P$ of measurements used to estimate each $X_P$.
The empirical estimator is
\begin{align}
    \widehat{X}_P
    &=
    \frac{
        \widetilde{P}_{\mathrm{PEC}}
    }{
        \operatorname{Tr}[P\rho]
    }
    \nonumber\\
    &=
    \frac{
        \widetilde{Q}
    }{
        m_P\operatorname{Tr}[P\rho]
    }
    \sum_{j=1}^{m_P}
    \operatorname{sign}
    \!\left[
        \widetilde{a}(\vec{P}'_j)
    \right]
    \mu_j(P),
    \qquad
    \vec{P}'_j\sim\widetilde{p},
    \label{eq:empirical_XP}
\end{align}
where $\mu_j(P)\in\{-1,1\}$ is the outcome of the
single-shot measurement of $P$.

We condition on a fixed realization of
$P_1,\ldots,P_\ell$. Under this conditioning,
$\overline{Y}$ is fixed and
\begin{equation}
    \mathbb{E}
    \!\left[
        \widehat{Y}
        \mid
        P_1,\ldots,P_\ell
    \right]
    =
    \overline{Y}.
\end{equation}
The individual random terms entering
$\ell\widehat{Y}$ are
\begin{equation}
    Z_{i,j}
    =
    \frac{
        \widetilde{Q}
    }{
        m_{P_i}\operatorname{Tr}[P_i\rho]
    }
    \operatorname{sign}
    \!\left[
        \widetilde{a}(\vec{P}'_{i,j})
    \right]
    \mu_{i,j}(P_i),
\end{equation}
and satisfy
\begin{equation}
    -\frac{
        \widetilde{Q}
    }{
        m_{P_i}
        \left|
            \operatorname{Tr}[P_i\rho]
        \right|
    }
    \leq
    Z_{i,j}
    \leq
    \frac{
        \widetilde{Q}
    }{
        m_{P_i}
        \left|
            \operatorname{Tr}[P_i\rho]
        \right|
    }.
\end{equation}
Hoeffding's inequality therefore yields
\begin{align}
    &\Pr\!\left[
        \left|
            \widehat{Y}-\overline{Y}
        \right|
        \geq
        \epsilon
        \,\middle|\,
        P_1,\ldots,P_\ell
    \right]
    \nonumber\\
    &\quad\leq
    2\exp\!\left[
        -
        \frac{
            \ell^2\epsilon^2
        }{
            \displaystyle
            \sum_{i=1}^{\ell}
            \frac{
                2\widetilde{Q}^{\,2}
            }{
                m_{P_i}
                \operatorname{Tr}[P_i\rho]^2
            }
        }
    \right]
    \nonumber\\
    &\quad=
    2\exp\!\left[
        -
        \frac{
            \ell^2\epsilon^2
        }{
            \displaystyle
            \sum_{i=1}^{\ell}
            \frac{
                2\widetilde{Q}^{\,2}
            }{
                m_{P_i}
                2^N\chi_\rho(P_i)^2
            }
        }
    \right].
    \label{eq:conditional_measurement_bound}
\end{align}
For every $P\in\mathcal{S}_\rho$, this probability is
upper-bounded by $\delta$ if
\begin{equation}
\label{eq:general_mP_bound}
    m_P
    \geq
    \frac{
        2\widetilde{Q}^{\,2}
    }{
        2^N\chi_\rho(P)^2
        \ell\epsilon^2
    }
    \ln\!\left(
        \frac{2}{\delta}
    \right).
\end{equation}
Because this conditional bound holds for every realization
of $P_1,\ldots,P_\ell$, it also holds without conditioning:
\begin{equation}
    \Pr\!\left[
        \left|
            \widehat{Y}-\overline{Y}
        \right|
        \geq
        \epsilon
    \right]
    \leq
    \delta.
\end{equation}

Combining the two sources of uncertainty using the triangle
inequality and the union bound~\cite{Kliesch_2021} gives
\begin{equation}
\label{bound precision Y}
    \Pr\!\left[
        \left|
            \widehat{Y}
            -
            \operatorname{Tr}(\widetilde{\rho}\rho)
        \right|
        \leq
        2\epsilon
    \right]
    \geq
    1-2\delta
\end{equation}
when
\begin{equation}
    \ell
    =
    \left\lceil
        \frac{
            \widetilde{Q}^{\,2}
        }{
            \epsilon^2\delta
        }
    \right\rceil,
    \qquad
    m_P
    =
    \left\lceil
        \frac{
            2\widetilde{Q}^{\,2}
        }{
            2^N\chi_\rho(P)^2
            \ell\epsilon^2
        }
        \ln\!\left(
            \frac{2}{\delta}
        \right)
    \right\rceil.
\end{equation}

Although $m_P$ is fixed once $P$ is specified, $m_{P_i}$
is a random variable because $P_i$ is sampled according to
$\chi_\rho(P)^2$. Its expectation satisfies
\begin{align}
    \mathbb{E}[m_P]
    &=
    \sum_{P\in\mathcal{S}_\rho}
    \chi_\rho(P)^2m_P
    \nonumber\\
    &\leq
    1
    +
    \sum_{P\in\mathcal{S}_\rho}
    \frac{
        2\widetilde{Q}^{\,2}
    }{
        2^N\ell\epsilon^2
    }
    \ln\!\left(
        \frac{2}{\delta}
    \right)
    \nonumber\\
    &=
    1
    +
    \frac{
        2\widetilde{Q}^{\,2}
        |\mathcal{S}_\rho|
    }{
        2^N\ell\epsilon^2
    }
    \ln\!\left(
        \frac{2}{\delta}
    \right)
    \nonumber\\
    &\leq
    1
    +
    \frac{
        2^{N+1}\widetilde{Q}^{\,2}
    }{
        \ell\epsilon^2
    }
    \ln\!\left(
        \frac{2}{\delta}
    \right),
\end{align}
where we used
$|\mathcal{S}_\rho|\leq4^N$.

The total expected number of single-shot measurements is
therefore
\begin{align}
    N_t
    &=
    \mathbb{E}
    \!\left[
        \sum_{i=1}^{\ell}
        m_{P_i}
    \right]
    =
    \ell\mathbb{E}[m_P]
    \nonumber\\
    &\leq
    \ell
    +
    \frac{
        2^{N+1}\widetilde{Q}^{\,2}
    }{
        \epsilon^2
    }
    \ln\!\left(
        \frac{2}{\delta}
    \right)
    \nonumber\\
    &\leq
    1
    +
    \frac{
        \widetilde{Q}^{\,2}
    }{
        \epsilon^2\delta
    }
    +
    \frac{
        2\cdot2^N\widetilde{Q}^{\,2}
    }{
        \epsilon^2
    }
    \ln\!\left(
        \frac{2}{\delta}
    \right).
\end{align}
For standard DFE between physical quantum states, setting
$\widetilde{Q}=1$ recovers
~\cite{Flammia_2011,Kliesch_2021}
\begin{equation}
    N_t
    \leq
    1
    +
    \frac{1}{\epsilon^2\delta}
    +
    \frac{
        2\cdot2^N
    }{
        \epsilon^2
    }
    \ln\!\left(
        \frac{2}{\delta}
    \right).
\end{equation}

\subsection{Well-conditioned ideal states}
\label{appendix DFE bounds well-conditioned}

We now specialize the preceding analysis to an
$\alpha$-well-conditioned ideal state, with
$\alpha\in(0,1]$. By definition, for every Pauli operator
$P$, either $\operatorname{Tr}[\rho P]=0$ or
\begin{equation}
    \left|
        \operatorname{Tr}[\rho P]
    \right|
    \geq
    \alpha
\end{equation}
~\cite{Flammia_2011}. Every Pauli operator sampled by DFE
belongs to $\mathcal{S}_\rho$ and therefore satisfies the
second condition.

For every physical state
$\rho'_{\mathrm{noisy}}(\vec{P}')$ and Pauli operator $P$,
\begin{equation}
    \left|
        \operatorname{Tr}
        \!\left[
            P\rho'_{\mathrm{noisy}}(\vec{P}')
        \right]
    \right|
    \leq
    1.
\end{equation}
The PEC decomposition consequently implies
\begin{align}
    \left|
        \operatorname{Tr}[P\widetilde{\rho}]
    \right|
    &=
    \left|
        \sum_{\vec{P}'}
        \widetilde{a}(\vec{P}')
        \operatorname{Tr}
        \!\left[
            P\rho'_{\mathrm{noisy}}(\vec{P}')
        \right]
    \right|
    \nonumber\\
    &\leq
    \sum_{\vec{P}'}
    \left|
        \widetilde{a}(\vec{P}')
    \right|
    =
    \widetilde{Q}.
\end{align}
Hence, for every sampled Pauli operator,
\begin{equation}
    |X_P|
    =
    \frac{
        \left|
            \operatorname{Tr}[P\widetilde{\rho}]
        \right|
    }{
        \left|
            \operatorname{Tr}[P\rho]
        \right|
    }
    \leq
    \frac{
        \widetilde{Q}
    }{
        \alpha
    }.
\end{equation}

The random variables $X_{P_i}$ therefore lie in the
interval
\begin{equation}
    -\frac{\widetilde{Q}}{\alpha}
    \leq
    X_{P_i}
    \leq
    \frac{\widetilde{Q}}{\alpha}.
\end{equation}
Applying Hoeffding's inequality to
$\overline{Y}$ gives
\begin{equation}
    \Pr\!\left[
        \left|
            \overline{Y}
            -
            \operatorname{Tr}(\widetilde{\rho}\rho)
        \right|
        \geq
        \epsilon
    \right]
    \leq
    2\exp\!\left[
        -
        \frac{
            \ell\epsilon^2\alpha^2
        }{
            2\widetilde{Q}^{\,2}
        }
    \right].
\end{equation}
This probability is at most $\delta$ whenever
\begin{equation}
\label{eq:well_conditioned_ell_bound}
    \ell
    \geq
    \frac{
        2\widetilde{Q}^{\,2}
    }{
        \epsilon^2\alpha^2
    }
    \ln\!\left(
        \frac{2}{\delta}
    \right).
\end{equation}

For the measurement uncertainty, the conditional
Hoeffding bound in
Eq.~\eqref{eq:conditional_measurement_bound}, together with
$|\operatorname{Tr}[P\rho]|\geq\alpha$, shows that it is
sufficient to choose
\begin{equation}
    m_P
    \geq
    \frac{
        2\widetilde{Q}^{\,2}
    }{
        \alpha^2\ell\epsilon^2
    }
    \ln\!\left(
        \frac{2}{\delta}
    \right).
\end{equation}
This bound is independent of the sampled Pauli operator.
Combining the two uncertainty bounds using the triangle
inequality and the union bound gives
\begin{equation}
    \Pr\!\left[
        \left|
            \widehat{Y}
            -
            \operatorname{Tr}(\widetilde{\rho}\rho)
        \right|
        \leq
        2\epsilon
    \right]
    \geq
    1-2\delta
\end{equation}
by choosing
\begin{equation}
    \ell
    =
    \left\lceil
        \frac{
            2\widetilde{Q}^{\,2}
        }{
            \epsilon^2\alpha^2
        }
        \ln\!\left(
            \frac{2}{\delta}
        \right)
    \right\rceil
\end{equation}
and
\begin{equation}
    m_P
    =
    \left\lceil
        \frac{
            2\widetilde{Q}^{\,2}
        }{
            \alpha^2\ell\epsilon^2
        }
        \ln\!\left(
            \frac{2}{\delta}
        \right)
    \right\rceil
    =
    1.
\end{equation}
The last equality follows because the expression inside the
ceiling is strictly positive and, by the chosen value of
$\ell$, is no greater than one.

Thus, for an $\alpha$-well-conditioned ideal state, the
total number of single-shot measurements is
\begin{equation}
    N_t
    =
    \ell
    =
    \left\lceil
        \frac{
            2\widetilde{Q}^{\,2}
        }{
            \epsilon^2\alpha^2
        }
        \ln\!\left(
            \frac{2}{\delta}
        \right)
    \right\rceil,
\end{equation}
and therefore
\begin{equation}
    N_t
    \leq
    1
    +
    \frac{
        2\widetilde{Q}^{\,2}
    }{
        \epsilon^2\alpha^2
    }
    \ln\!\left(
        \frac{2}{\delta}
    \right).
\end{equation}
For constant $\alpha$, this bound has no explicit
dependence on the number of qubits.:

\section{pPEC grouping for Matchgate circuits}
\label{appendix ppec MG}

Error-propagated probabilistic error cancellation
(pPEC)~\cite{Scheiber2025reducedsampling} reduces the PEC
sampling factor by propagating Pauli corrections through a
circuit and combining correction paths that produce the
same physical circuit. The construction is particularly
natural for Clifford circuits, since Clifford conjugation
preserves the Pauli group. Here, we show that the same
grouping principle can be applied to Matchgate (MG)
circuits under a noise-invariance assumption.

Using the notation of the main text, the effective
mitigated output produced by standard PEC is
\begin{equation}
\label{eq: appendix final unitary}
    \widetilde{\rho}
    =
    \sum_{\vec{P}'}
    \widetilde{a}(\vec{P}')
    \rho'_{\mathrm{noisy}}(\vec{P}'),
\end{equation}
where
$\vec{P}'=(P'_1,\ldots,P'_s)$ specifies the Pauli
corrections inserted at the different circuit layers. The
corresponding standard PEC sampling factor is
\begin{equation}
    \widetilde{Q}
    =
    \sum_{\vec{P}'}
    \left|
        \widetilde{a}(\vec{P}')
    \right|.
\end{equation}

Consider the ideal unitary underlying one modified MG
circuit,
\begin{equation}
    U'(\vec{P}')
    =
    P'_sM_s\cdots P'_2M_2P'_1M_1,
\end{equation}
where $M_1$ acts first. A two-qubit MG connected
continuously to the identity can be written as
\begin{equation}
    M_k
    =
    e^{iH_k},
    \qquad
    H_k
    =
    \sum_{A\in\mathcal{G}_{\mathrm{MG}}}
    \theta_{k,A}A,
\end{equation}
where
\begin{equation}
    \mathcal{G}_{\mathrm{MG}}
    =
    \left\{
        Z_i,\,
        Z_j,\,
        X_iX_j,\,
        X_iY_j,\,
        Y_iX_j,\,
        Y_iY_j
    \right\}.
\end{equation}

Let $R\in\mathcal{P}_N$ be a phase-free Pauli string.
Because $R^2=I$, it can be propagated through an MG
according to
\begin{align}
\label{eq: no conmute MG P}
    RM_k
    &=
    \left(RM_kR\right)R
    \nonumber\\
    &=
    M_k^{(R)}R,
\end{align}
where
\begin{equation}
    M_k^{(R)}
    =
    \exp\!\left[
        i
        \sum_{A\in\mathcal{G}_{\mathrm{MG}}}
        s_{A,R}\theta_{k,A}A
    \right],
    \qquad
    RAR=s_{A,R}A,
\end{equation}
with $s_{A,R}\in\{-1,+1\}$. Thus, Pauli propagation
preserves the MG structure and only changes signs of its
Hamiltonian coefficients.

For an insertion string $\vec{P}'$, define the cumulative
phase-free Pauli
\begin{equation}
    R_k(\vec{P}')
    =
    P'_sP'_{s-1}\cdots P'_k.
\end{equation}
Repeated application of
Eq.~\eqref{eq: no conmute MG P} gives
\begin{equation}
    U'(\vec{P}')
    =
    M'_s(\vec{P}')
    \cdots
    M'_1(\vec{P}')
    P_T(\vec{P}'),
\end{equation}
up to an irrelevant global phase, where
\begin{equation}
    M'_k(\vec{P}')
    =
    M_k^{(R_k(\vec{P}'))},
    \qquad
    P_T(\vec{P}')
    =
    R_1(\vec{P}').
\end{equation}
Each PEC correction path therefore determines a Pauli
correction at the circuit input and a sequence of
sign-modified MGs.

For the computational-basis product inputs used in the
protocol, the propagated input Pauli can be reduced at the
level of the input density operator. On each qubit,
$I$ and $Z$ have the same action up to a phase, as do $X$
and $Y$. We may therefore apply the reduction
\begin{equation}
    I,Z\longmapsto I,
    \qquad
    X,Y\longmapsto X
\end{equation}
and denote the resulting string by $P_T^{XI}$. For any
computational-basis product state $\rho_0$,
\begin{equation}
    P_T\rho_0P_T^\dagger
    =
    P_T^{XI}\rho_0
    \left(P_T^{XI}\right)^\dagger.
\end{equation}

To extend this rewriting to the noisy circuits, we assume
that the Pauli-stochastic noise is covariant under Pauli
conjugation and that the noise associated with an MG is
invariant under the parameter-sign changes generated by
Eq.~\eqref{eq: no conmute MG P}. The crosstalk and
over-rotation models considered in this manuscript satisfy
this assumption. The crosstalk channel does not depend on
the MG angles. For the over-rotation model,
\begin{equation}
    H_k^{(\mathrm{OR})}
    =
    \gamma
    \sum_P
    |\beta_{k,P}|P,
\end{equation}
and a propagated sign change
$\beta_{k,P}\mapsto s_{P,R}\beta_{k,P}$ leaves
$|\beta_{k,P}|$ unchanged.

Under these assumptions, each insertion string determines
the equivalence-class label
\begin{equation}
    g(\vec{P}')
    =
    \left(
        P_T^{XI}(\vec{P}'),
        \left\{
            M'_k(\vec{P}')
        \right\}_{k=1}^{s}
    \right).
\end{equation}
Two insertion strings with the same label produce the same
propagated noisy output state. Defining the grouped
coefficient
\begin{equation}
    \widetilde{A}_g
    =
    \sum_{\vec{P}'\,:\,g(\vec{P}')=g}
    \widetilde{a}(\vec{P}'),
\end{equation}
the effective mitigated output can be rewritten as
\begin{equation}
\label{eq: appendix final unitary grouped MG}
    \widetilde{\rho}
    =
    \sum_g
    \widetilde{A}_g
    \rho'_{\mathrm{noisy}}(g).
\end{equation}
The sampling factor of this grouped decomposition satisfies
\begin{align}
    \widetilde{Q}_{\mathrm{pPEC}}
    &=
    \sum_g
    \left|
        \widetilde{A}_g
    \right|
    \nonumber\\
    &=
    \sum_g
    \left|
        \sum_{\vec{P}'\,:\,g(\vec{P}')=g}
        \widetilde{a}(\vec{P}')
    \right|
    \nonumber\\
    &\leq
    \sum_g
    \sum_{\vec{P}'\,:\,g(\vec{P}')=g}
    \left|
        \widetilde{a}(\vec{P}')
    \right|
    \nonumber\\
    &=
    \widetilde{Q}.
\end{align}

This inequality establishes only that the exact grouping
cannot increase the sampling factor; it does not guarantee
a strict or practically significant reduction for a given
MG circuit. Equality holds, for example, if no two
correction paths produce the same propagated circuit or if
all coefficients combined within each equivalence class
have the same sign. Moreover, constructing all grouped
coefficients may require enumerating exponentially many
correction paths. Blockwise grouping can trade classical
preprocessing for a potentially smaller sampling factor,
as discussed in
Ref.~\cite{Scheiber2025reducedsampling}, but we do not
analyze this trade-off or the achievable reduction
numerically. Accordingly, the result established here is
the algebraic grouping rule and the bound
$\widetilde{Q}_{\mathrm{pPEC}}\leq\widetilde{Q}$. All
numerical experiments reported in this manuscript use
standard PEC and do not apply this or any other
sampling-reduction strategy.

\section{Effective PEC mitigated state for the global depolarizing error model}
\label{appendix global depolarizing mitigated state}

Here we derive the effective mitigated state obtained when
PEC is applied to a circuit affected by global depolarizing
noise after every gate. Define the completely depolarizing
channel
\begin{equation}
    \mathcal{T}(\rho)
    =
    \frac{1}{4^N}
    \sum_{P\in\mathcal{P}_N}P\rho P
    =
    \operatorname{Tr}(\rho)\frac{I}{2^N}.
\end{equation}
The global depolarizing channel with error parameter $p_e$
can then be written as
\begin{equation}
    \mathcal{D}_{p_e}
    =
    (1-p_e)\mathcal{I}+p_e\mathcal{T}.
\end{equation}

PEC is implemented using an estimated error parameter
$\widetilde p_e$. For $0\leq\widetilde p_e<1$, the estimated
channel is invertible, and its inverse map is
\begin{equation}
    \mathcal{D}_{\widetilde p_e}^{-1}
    =
    \frac{1}{1-\widetilde p_e}
    \left(
        \mathcal{I}-\widetilde p_e\mathcal{T}
    \right),
\end{equation}
where we used $\mathcal{T}^2=\mathcal{T}$. Composing this
inverse map with the true noise channel gives
\begin{equation}
\begin{split}
    \mathcal{D}_{\widetilde p_e}^{-1}
    \circ\mathcal{D}_{p_e}
    &=
    (1-p_{\mathrm{PEC}})\mathcal{I}
    +p_{\mathrm{PEC}}\mathcal{T},\\
    p_{\mathrm{PEC}}
    &=
    \frac{p_e-\widetilde p_e}
         {1-\widetilde p_e}.
\end{split}
\end{equation}
Thus, imperfect PEC produces a residual map of the same
global-depolarizing form. When $\widetilde p_e=p_e$, we
have $p_{\mathrm{PEC}}=0$ and recover the ideal operation.
For $\widetilde p_e<p_e$, the residual map is a physical
global depolarizing channel. For $\widetilde p_e>p_e$,
$p_{\mathrm{PEC}}<0$, and the residual operation is a
trace-preserving but non-completely-positive linear map.

Since the completely depolarizing channel commutes with
every unitary channel and satisfies $\mathcal{T}^2=\mathcal{T}$,
after $s$ noisy gates the final effective mitigated state is
\begin{equation}
    \widetilde{\rho}
    =
    (1-p_{\mathrm{PEC}})^s\rho
    +
    \left[
        1-(1-p_{\mathrm{PEC}})^s
    \right]
    \frac{I}{2^N}.
\end{equation}

\section{Pauli channel coefficients for the crosstalk and over-rotation noise channels}
\label{appendix pauli coefficients crosstalk plus overotaion}

In this appendix, we derive the Pauli coefficients and
the corresponding Pauli-fidelity conditions for the
crosstalk and over-rotation noise models considered in
Sec.~\ref{sec: applications}. For a stochastic Pauli
channel, the Pauli fidelity associated with $P$ is related
to the Pauli coefficients by
\begin{align}
    f_{k,P}
    &=
    c_k(I)
    +
    \sum_{\substack{[P',P]=0\\P'\neq I}}
    c_k(P')
    -
    \sum_{\{P',P\}=0}
    c_k(P')
    \notag\\
    &=
    1
    -
    2
    \sum_{\{P',P\}=0}
    c_k(P').
    \label{eq:appendix-F-fidelities}
\end{align}

The crosstalk noise channel characterized by
\begin{align}
    \mathcal{C}_{ij}(\rho)
    &= (1-p_c)\rho
    +\frac{p_c}{4}\Big(
    X_iX_j\rho X_iX_j
    +X_iY_j\rho X_iY_j
    \notag\\
    &\qquad\qquad
    +Y_iX_j\rho Y_iX_j
    +Y_iY_j\rho Y_iY_j
    \Big)
\end{align}
is already a Pauli-stochastic channel. Its nonzero Pauli
coefficients are
\begin{align}
    c(I)&=1-p_c,\notag\\
    c(X_iX_j)
    &=c(X_iY_j)
    =c(Y_iX_j)
    =c(Y_iY_j)
    =\frac{p_c}{4}.
\end{align}
All other nonidentity Pauli coefficients are zero.
Therefore, $\widetilde p_c<p_c$ satisfies the coefficientwise
underestimation definition of
Sec.~\ref{sec: directionality general}, whereas
$\widetilde p_c>p_c$ satisfies the corresponding
overestimation definition.

Within the positive-fidelity regime stated in
Sec.~\ref{sec: directionality general},
Eq.~\eqref{eq:appendix-F-fidelities} then gives
\begin{equation}
    0<f_{k,P}\leq\widetilde f_{k,P}
    \qquad
    \text{for all }k,P
    \label{eq:appendix-F-crosstalk-under}
\end{equation}
when $\widetilde p_c<p_c$, with at least one strict
inequality. Conversely, when $\widetilde p_c>p_c$,
\begin{equation}
    0<\widetilde f_{k,P}\leq f_{k,P}
    \qquad
    \text{for all }k,P,
    \label{eq:appendix-F-crosstalk-over}
\end{equation}
again with at least one strict inequality.

We next consider the coherent over-rotation error defined
by
\begin{equation}
    O_kG_k
    =
    e^{iH_k^{(\mathrm{OR})}}
    e^{iH_k},
\end{equation}
where the noiseless generating Hamiltonian is
\begin{equation}
    H_k
    =
    \sum_P
    \beta_{k,P}P,
\end{equation}
and the coherent over-rotation is generated by
\begin{equation}
    H_k^{(\mathrm{OR})}
    =
    \gamma
    \sum_P
    |\beta_{k,P}|P.
\end{equation}
Any identity component of $H_k$ can be omitted because it
contributes only a global phase. We therefore define
\begin{equation}
    A_k
    =
    \sum_{P\neq I}
    |\beta_{k,P}|P,
\end{equation}
so that
\begin{equation}
    O_k
    =
    e^{i\gamma A_k}.
\end{equation}

After applying Pauli twirling~\cite{van_den_Berg_2023,Wallman_2016},
the corresponding Pauli coefficients are
\begin{equation}
    c_k(P)
    =
    \frac{1}{4^N}
    \left|
        \operatorname{Tr}(PO_k)
    \right|^2.
\end{equation}
To determine their behavior for small positive $\gamma$,
we expand
\begin{equation}
    O_k
    =
    I
    +
    i\gamma A_k
    -
    \frac{\gamma^2}{2}A_k^2
    +
    \mathcal{O}(\gamma^3).
\end{equation}

For the identity coefficient,
$\operatorname{Tr}(A_k)=0$, and the orthogonality of the
Pauli basis gives
\begin{equation}
    \operatorname{Tr}(A_k^2)
    =
    2^N
    \sum_{P\neq I}
    |\beta_{k,P}|^2.
\end{equation}
Consequently,
\begin{align}
    \operatorname{Tr}(O_k)
    &=
    2^N
    -
    \frac{\gamma^2}{2}
    \operatorname{Tr}(A_k^2)
    +
    \mathcal{O}(\gamma^3)
    \notag\\
    &=
    2^N
    \left(
        1
        -
        \frac{\gamma^2}{2}
        \sum_{P\neq I}
        |\beta_{k,P}|^2
    \right)
    +
    \mathcal{O}(\gamma^3).
\end{align}
It follows that
\begin{equation}
    c_k(I)
    =
    1
    -
    \gamma^2
    \sum_{P\neq I}
    |\beta_{k,P}|^2
    +
    \mathcal{O}(\gamma^4).
\end{equation}
Thus, for a nontrivial over-rotation generator, the
identity coefficient decreases with $\gamma$ in a
sufficiently small neighborhood of $\gamma=0$.

For a nonidentity Pauli operator $P$, we use the full
expansion
\begin{equation}
    \operatorname{Tr}(PO_k)
    =
    \sum_{r=0}^{\infty}
    \frac{(i\gamma)^r}{r!}
    \operatorname{Tr}(PA_k^r).
\end{equation}
Since $\operatorname{Tr}(P)=0$ for $P\neq I$, the
zeroth-order term vanishes. Let $t_P\geq1$ be the
smallest integer for which
\begin{equation}
    \operatorname{Tr}(PA_k^{t_P})
    \neq0.
\end{equation}
Then
\begin{equation}
    \operatorname{Tr}(PO_k)
    =
    \frac{(i\gamma)^{t_P}}{t_P!}
    \operatorname{Tr}(PA_k^{t_P})
    +
    \mathcal{O}(\gamma^{t_P+1}),
\end{equation}
and hence
\begin{equation}
    c_k(P)
    =
    \frac{\gamma^{2t_P}}{4^N}
    \left|
        \frac{
            \operatorname{Tr}(PA_k^{t_P})
        }{
            t_P!
        }
    \right|^2
    +
    \mathcal{O}(\gamma^{2t_P+1}).
\end{equation}
The leading coefficient is positive, so $c_k(P)$
increases with $\gamma$ for sufficiently small positive
$\gamma$. If no such integer $t_P$ exists, then
$\operatorname{Tr}(PO_k)=0$ for all $\gamma$, and the
corresponding Pauli coefficient remains zero.

Therefore, in a sufficiently small neighborhood of
$\gamma=0$, every nonidentity coefficient is a
nondecreasing function of $\gamma$, and every coefficient
that appears at a nonzero leading order is strictly
increasing. Thus, $\widetilde\gamma<\gamma$ satisfies the
coefficientwise underestimation definition, whereas
$\widetilde\gamma>\gamma$ satisfies the coefficientwise
overestimation definition.

Using Eq.~\eqref{eq:appendix-F-fidelities}, within the
perturbative positive-fidelity regime,
$\widetilde\gamma<\gamma$ implies
\begin{equation}
    0<f_{k,P}\leq\widetilde f_{k,P}
    \qquad
    \text{for all }k,P,
    \label{eq:appendix-F-overrotation-under}
\end{equation}
with at least one strict inequality for a nontrivial
over-rotation. Conversely, $\widetilde\gamma>\gamma$ implies
\begin{equation}
    0<\widetilde f_{k,P}\leq f_{k,P}
    \qquad
    \text{for all }k,P,
    \label{eq:appendix-F-overrotation-over}
\end{equation}
again with at least one strict inequality. For larger
values of $\gamma$, the exact coefficients are
trigonometric functions and are not generally monotonic.

The conditions derived above apply separately to the
crosstalk and over-rotation channels as constituent error
sources. When the two channels are concatenated,
$\widetilde p_c<p_c$ together with
$\widetilde\gamma<\gamma$ constitutes consistent
underestimation, while $\widetilde p_c>p_c$ together with
$\widetilde\gamma>\gamma$ constitutes consistent
overestimation. Since the Pauli fidelities multiply under
concatenation, these cases satisfy
Eqs.~\eqref{eq: fidelities underestimation} and
\eqref{eq: fidelities overestimation}, respectively.

We do not infer that the Pauli coefficients of the
concatenated channel satisfy the same coefficientwise
ordering. As explained in Sec.~\ref{sec: ptm general},
the extension of the analytical DFE-overlap and DFE-purity
results to the concatenated model instead follows from
the multiplicativity of the Pauli fidelities. No
corresponding conclusion holds for hybrid
miscalibration.

\section{Upper bounding the DFE overlap in the underestimation case}
\label{app:underestimation_overlap}

\begin{lemma}[Underestimation bound]
Consider an arbitrary unitary circuit acting on a pure
initial state $\rho_0$, and let $\rho$ and
$\widetilde{\rho}$ denote its ideal and effective mitigated
output states, respectively. Assume that every residual map
is Pauli diagonal and that its effective Pauli fidelities
satisfy
\begin{equation}
    0
    <
    f_{k,P,\mathrm{eff}}
    =
    \frac{f_{k,P}}{\widetilde{f}_{k,P}}
    \leq
    1
    \qquad
    \forall\, k,P.
    \label{eq:effective_fidelities_under_app}
\end{equation}
Then
\begin{equation}
    \operatorname{Tr}(\widetilde{\rho}^2)
    \leq
    1
\end{equation}
and
\begin{equation}
    \left|
        \operatorname{Tr}(\widetilde{\rho}\rho)
    \right|
    \leq
    1.
\end{equation}
In particular,
\begin{equation}
    \operatorname{Tr}(\widetilde{\rho}\rho)
    \leq
    1.
\end{equation}
\end{lemma}

\begin{proof}
Let $\widetilde{\rho}_{k-1}$ denote the effective mitigated
state after layer $k-1$, and define the intermediate state
\begin{equation}
    \sigma_k
    =
    G_k\widetilde{\rho}_{k-1}G_k^\dagger,
    \label{eq:sigma_under_app}
\end{equation}
which corresponds to applying the $k$-th ideal gate before
the $k$-th effective map is applied. The effective
mitigated state after layer $k$ is
\begin{equation}
    \widetilde{\rho}_k
    =
    \mathcal{D}_{k,\mathrm{eff}}(\sigma_k).
    \label{eq:rho_tilde_under_app}
\end{equation}
Since the effective mitigated map is diagonal in the
Pauli basis,
\begin{equation}
    \operatorname{Tr}(P\widetilde{\rho}_k)
    =
    f_{k,P,\mathrm{eff}}
    \operatorname{Tr}(P\sigma_k).
    \label{eq:pauli_components_under_app}
\end{equation}
It follows that
\begin{align}
    \operatorname{Tr}(\widetilde{\rho}_k^2)
    &=
    \frac{1}{2^N}
    \sum_{P\in\mathcal{P}_N}
    f_{k,P,\mathrm{eff}}^2
    \operatorname{Tr}(P\sigma_k)^2
    \nonumber\\
    &\leq
    \frac{1}{2^N}
    \sum_{P\in\mathcal{P}_N}
    \operatorname{Tr}(P\sigma_k)^2
    \nonumber\\
    &=
    \operatorname{Tr}(\sigma_k^2)
    =
    \operatorname{Tr}(\widetilde{\rho}_{k-1}^2),
    \label{eq:purity_contraction_under_app}
\end{align}
where the inequality follows from
Eq.~\eqref{eq:effective_fidelities_under_app}, and the last
equality follows from the cyclic property of the trace.

Thus, the DFE purity of the effective mitigated state cannot
increase after any layer satisfying
Eq.~\eqref{eq:effective_fidelities_under_app}. For a pure
initial state, $\widetilde{\rho}_0=\rho_0$ and
$\operatorname{Tr}(\rho_0^2)=1$, and hence
\begin{equation}
    \operatorname{Tr}(\widetilde{\rho}^2)
    \leq
    1.
    \label{eq:purity_upper_under_app}
\end{equation}

The ideal output state $\rho$ is pure and therefore
satisfies $\operatorname{Tr}(\rho^2)=1$. Applying the
Hilbert--Schmidt Cauchy--Schwarz inequality gives
\begin{align}
    \left|
        \operatorname{Tr}(\widetilde{\rho}\rho)
    \right|
    &\leq
    \sqrt{
        \operatorname{Tr}(\rho^2)
        \operatorname{Tr}(\widetilde{\rho}^2)
    }
    \nonumber\\
    &\leq
    1.
    \label{eq:overlap_upper_under_app}
\end{align}
We therefore conclude that
\begin{equation}
    \operatorname{Tr}(\widetilde{\rho}\rho)
    \leq
    1
\end{equation}
for arbitrary unitary circuits whenever the effective
Pauli fidelities satisfy
Eq.~\eqref{eq:effective_fidelities_under_app}.
\end{proof}

\section{Lower bounding the DFE overlap in the overestimation case for Clifford circuits}
\label{appendix: clifford proof}

As discussed in the main manuscript and in
Appendix~\ref{appendix pauli coefficients crosstalk plus overotaion},
the effective-fidelity condition below holds for an
individually overestimated error source and is preserved
when several constituent error sources are consistently
overestimated.

\begin{lemma}[Overestimation bound for Clifford circuits]
Let the ideal circuit act on a pure initial state, and
assume that every residual map is Pauli diagonal with
effective Pauli fidelities satisfying
\begin{equation}
    f_{k,P,\mathrm{eff}}
    =
    \frac{f_{k,P}}{\widetilde{f}_{k,P}}
    \geq
    1
    \qquad
    \forall\,k,P.
\end{equation}
For a one-layer circuit,
\begin{equation}
    \operatorname{Tr}
    (\widetilde{\rho}_1\rho_1)
    \geq
    1
\end{equation}
without any assumption that the ideal gate is Clifford.
For a circuit with an arbitrary number of layers, if every
ideal layer is Clifford, then
\begin{equation}
    \operatorname{Tr}(\widetilde{\rho}\rho)
    \geq
    1.
\end{equation}
\end{lemma}

\begin{proof}
Using the PTM formalism, any state $\rho$ can be
expressed in terms of its Pauli components as
\begin{equation}
    |\rho\rangle\rangle_P
    =
    \frac{1}{2^N}
    \operatorname{Tr}(P\rho),
\end{equation}
where $|\rho\rangle\rangle$ denotes the vectorized form
of the density matrix. The action of any linear map on
$\rho$ is then represented as a matrix--vector
multiplication between the PTM of the map and the
vectorized state~\cite{Roncallo_2023,Hantzko_2025}.

Using this notation, the effective mitigated state after
the first noisy layer, $k=1$, can be written as
\begin{equation}
    |\widetilde{\rho}_1\rangle\rangle_P
    =
    f_{1,P,\mathrm{eff}}
    \frac{1}{2^N}
    \operatorname{Tr}(P\rho_1)
    =
    f_{1,P,\mathrm{eff}}
    |\rho_1\rangle\rangle_P,
\end{equation}
where $\rho_1$ is the ideal state obtained after applying
the first layer of noiseless gates. We can now compute the
DFE overlap between the effective mitigated and ideal
states after the first layer:
\begin{align}
    \operatorname{Tr}(\widetilde{\rho}_1\rho_1)
    &=
    2^N
    \langle\langle\rho_1
    |
    \widetilde{\rho}_1
    \rangle\rangle
    \nonumber\\
    &=
    \frac{2^N}{4^N}
    \sum_{P\in\mathcal{P}_N}
    f_{1,P,\mathrm{eff}}
    \operatorname{Tr}(P\rho_1)^2
    \nonumber\\
    &\geq
    \frac{2^N}{4^N}
    \sum_{P\in\mathcal{P}_N}
    \operatorname{Tr}(P\rho_1)^2
    \nonumber\\
    &=
    \operatorname{Tr}(\rho_1^2)
    =
    1,
\end{align}
where we used $f_{1,P,\mathrm{eff}}\geq1$ for every $P$
and the fact that the ideal state $\rho_1$ is pure.

We have not assumed that the ideal state $\rho_1$ is
generated by a Clifford circuit. Therefore, for one-layer
circuits, the DFE overlap is lower-bounded by one whenever
the effective-fidelity condition above holds. We now
extend the proof to multiple layers under the assumption
that the circuit is composed of Clifford gates. The
argument proceeds by induction on the number of layers.

Assume that, for some layer $k$,
\begin{equation}
    \operatorname{Tr}(\rho_k\widetilde{\rho}_k)
    \geq
    1.
\end{equation}
We define
\begin{align}
    \rho_{k+1}
    &=
    C_{k+1}\rho_kC_{k+1}^\dagger,
    \\
    \widetilde{\rho}_{k+1}
    &=
    \mathcal{D}_{k+1,\mathrm{eff}}
    \left(
        C_{k+1}\widetilde{\rho}_kC_{k+1}^\dagger
    \right).
\end{align}
where $C_{k+1}$ is the Clifford gate in the
$(k+1)$-th layer and
$\mathcal{D}_{k+1,\mathrm{eff}}$ is the corresponding
effective Pauli-diagonal map. We also introduce the
intermediate state
\begin{equation}
    \sigma_{k+1}
    =
    C_{k+1}
    \widetilde{\rho}_k
    C_{k+1}^\dagger.
\end{equation}
It follows that
\begin{align}
    \operatorname{Tr}(\sigma_{k+1}\rho_{k+1})
    &=
    \operatorname{Tr}
    \!\left(
        C_{k+1}\widetilde{\rho}_kC_{k+1}^\dagger
        C_{k+1}\rho_kC_{k+1}^\dagger
    \right)
    \nonumber\\
    &=
    \operatorname{Tr}(\widetilde{\rho}_k\rho_k)
    \geq
    1,
\end{align}
where we used cyclicity of the trace and the induction
hypothesis.

In PTM notation, this becomes
\begin{equation}
    \operatorname{Tr}(\sigma_{k+1}\rho_{k+1})
    =
    2^N
    \sum_{P\in\mathcal{P}_N}
    |\sigma_{k+1}\rangle\rangle_P
    |\rho_{k+1}\rangle\rangle_P
    \geq
    1.
\end{equation}
Since the effective map at layer $k+1$ is diagonal in the
Pauli basis,
\begin{equation}
    |\widetilde{\rho}_{k+1}\rangle\rangle_P
    =
    f_{k+1,P,\mathrm{eff}}
    |\sigma_{k+1}\rangle\rangle_P.
\end{equation}
Therefore,
\begin{equation}
    \operatorname{Tr}
    (\widetilde{\rho}_{k+1}\rho_{k+1})
    =
    2^N
    \sum_{P\in\mathcal{P}_N}
    f_{k+1,P,\mathrm{eff}}
    |\sigma_{k+1}\rangle\rangle_P
    |\rho_{k+1}\rangle\rangle_P.
\end{equation}

In contrast to the one-layer case, we cannot immediately
lower-bound this expression by removing the factors
$f_{k+1,P,\mathrm{eff}}$, since the summands can, in
general, have different signs. The key step is to show
that, for Clifford circuits,
\begin{equation}
    |\sigma_{k+1}\rangle\rangle_P
    |\rho_{k+1}\rangle\rangle_P
    \geq
    0
    \qquad
    \forall\,P.
\end{equation}
It then follows that
\begin{align}
    \operatorname{Tr}
    (\widetilde{\rho}_{k+1}\rho_{k+1})
    &=
    2^N
    \sum_{P\in\mathcal{P}_N}
    f_{k+1,P,\mathrm{eff}}
    |\sigma_{k+1}\rangle\rangle_P
    |\rho_{k+1}\rangle\rangle_P
    \nonumber\\
    &\geq
    2^N
    \sum_{P\in\mathcal{P}_N}
    |\sigma_{k+1}\rangle\rangle_P
    |\rho_{k+1}\rangle\rangle_P
    \nonumber\\
    &\geq
    1.
\end{align}

It remains to prove the componentwise sign condition. We
show this by induction. For the base case,
$\widetilde{\rho}_0=\rho_0$ implies
$\sigma_1=\rho_1$, and therefore
\begin{equation}
    |\sigma_1\rangle\rangle_P
    |\rho_1\rangle\rangle_P
    \geq
    0
    \qquad
    \forall\,P.
\end{equation}
For the induction step, assume that, for some $k$,
\begin{equation}
    |\sigma_k\rangle\rangle_P
    |\rho_k\rangle\rangle_P
    \geq
    0
    \qquad
    \forall\,P,
\end{equation}
or equivalently,
\begin{equation}
    \operatorname{Tr}(P\sigma_k)
    \operatorname{Tr}(P\rho_k)
    \geq
    0
    \qquad
    \forall\,P.
\end{equation}
Since
\begin{equation}
    |\widetilde{\rho}_k\rangle\rangle_P
    =
    f_{k,P,\mathrm{eff}}
    |\sigma_k\rangle\rangle_P
\end{equation}
and $f_{k,P,\mathrm{eff}}\geq1$, it follows that
\begin{equation}
    \operatorname{Tr}(P\widetilde{\rho}_k)
    \operatorname{Tr}(P\rho_k)
    \geq
    0
    \qquad
    \forall\,P.
\end{equation}

We now examine the states after applying the Clifford gate
$C_{k+1}$:
\begin{align}
    &\operatorname{Tr}
    \!\left(
        C_{k+1}
        \widetilde{\rho}_k
        C_{k+1}^\dagger
        P
    \right)
    \operatorname{Tr}
    \!\left(
        C_{k+1}
        \rho_k
        C_{k+1}^\dagger
        P
    \right)
    \nonumber\\
    &=
    \operatorname{Tr}
    \!\left(
        \widetilde{\rho}_k
        C_{k+1}^\dagger
        PC_{k+1}
    \right)
    \operatorname{Tr}
    \!\left(
        \rho_k
        C_{k+1}^\dagger
        PC_{k+1}
    \right).
\end{align}
Because $C_{k+1}$ is a Clifford gate, there exist a
phase-free Pauli operator $P'\in\mathcal{P}_N$ and a sign
$s_P\in\{-1,1\}$ such that
\begin{equation}
    C_{k+1}^\dagger
    PC_{k+1}
    =
    s_PP'.
\end{equation}
Consequently,
\begin{align}
    &\operatorname{Tr}
    \!\left(
        \widetilde{\rho}_k
        C_{k+1}^\dagger
        PC_{k+1}
    \right)
    \operatorname{Tr}
    \!\left(
        \rho_k
        C_{k+1}^\dagger
        PC_{k+1}
    \right)
    \nonumber\\
    &=
    s_P^2
    \operatorname{Tr}(\widetilde{\rho}_kP')
    \operatorname{Tr}(\rho_kP')
    \nonumber\\
    &=
    \operatorname{Tr}(\widetilde{\rho}_kP')
    \operatorname{Tr}(\rho_kP')
    \geq
    0,
\end{align}
where the final inequality follows from the induction
hypothesis. Thus,
\begin{equation}
    |\sigma_{k+1}\rangle\rangle_P
    |\rho_{k+1}\rangle\rangle_P
    \geq
    0
    \qquad
    \forall\,P,
\end{equation}
which completes the induction.

We therefore conclude that, for Clifford circuits,
\begin{equation}
    \operatorname{Tr}(\widetilde{\rho}\rho)
    \geq
    1
\end{equation}
whenever the effective Pauli fidelities satisfy
$f_{k,P,\mathrm{eff}}\geq1$ for every $k$ and $P$. In
particular, this includes concatenated noise models for
which all constituent error sources are consistently
overestimated.
\end{proof}

\section{The DFE overlap asymmetric response beyond Clifford circuits}
\label{app:beyond-clifford}

For non-Clifford circuits, it is useful to express the
effective mitigated state in terms of the coefficients of
the effective Pauli-diagonal maps, following the notation
of Ref.~\cite{carrasco2024}. For notational simplicity, we
regard each layer map as an $N$-qubit map by extending it
with the identity outside its support. The effective map
at layer $k$ can then be written as
\begin{equation}
    \mathcal{D}_{k,\mathrm{eff}}(\cdot)
    =
    \sum_{P_k\in\mathcal{P}_N}
    c_{k,\mathrm{eff}}(P_k)\,
    P_k(\cdot)P_k.
\end{equation}
Let
\begin{equation}
    \rho_0
    =
    \ketbra{\psi_0}{\psi_0}
\end{equation}
be an arbitrary pure initial state. The effective mitigated
state admits the trajectory expansion
\begin{equation}
    \widetilde{\rho}
    =
    \sum_{\vec{P}}
    c_{\mathrm{eff}}(\vec{P})\,
    \ketbra*{\psi(\vec{P})}{\psi(\vec{P})},
\end{equation}
where
\begin{equation}
    c_{\mathrm{eff}}(\vec{P})
    =
    \prod_{k=1}^{s}c_{k,\mathrm{eff}}(P_k),
\end{equation}
and
\begin{equation}
    \ket*{\psi(\vec{P})}
    =
    P_sG_s\cdots P_2G_2P_1G_1
    \ket{\psi_0}.
\end{equation}
For $\vec{I}=(I,\ldots,I)$, the corresponding trajectory
is the ideal output state,
\begin{equation}
    \rho
    =
    \ketbra*{\psi(\vec{I})}{\psi(\vec{I})}.
\end{equation}
The DFE overlap is therefore
\begin{equation}
    \label{eq: fidelity pauli coefficients}
    \operatorname{Tr}(\widetilde{\rho}\rho)
    =
    \sum_{\vec{P}}
    c_{\mathrm{eff}}(\vec{P})\,
    \left|
        \braket*{\psi(\vec{I})}{\psi(\vec{P})}
    \right|^2.
\end{equation}

The coefficient associated with the identity trajectory
satisfies
\begin{equation}
    c_{\mathrm{eff}}(\vec{I})
    =
    \prod_{k=1}^{s}c_{k,\mathrm{eff}}(I),
\end{equation}
where
\begin{equation}
    c_{k,\mathrm{eff}}(I)
    =
    \frac{1}{4^N}
    \sum_{P\in\mathcal{P}_N}
    f_{k,P,\mathrm{eff}}.
\end{equation}
It follows that
\begin{equation}
    \label{eq: c_eff directionality}
    \begin{aligned}
        c_{\mathrm{eff}}(\vec{I})
        &\geq 1,
        &&\text{if }
        f_{k,P,\mathrm{eff}}\geq1
        \quad\forall\,k,P,\\
        c_{\mathrm{eff}}(\vec{I})
        &\leq 1,
        &&\text{if }
        0<f_{k,P,\mathrm{eff}}\leq1
        \quad\forall\,k,P.
    \end{aligned}
\end{equation}
Thus, the identity-trajectory coefficient exhibits the
expected asymmetric response. When several constituent
error sources are concatenated, the effective-fidelity
inequalities in Eq.~\eqref{eq: c_eff directionality} are
preserved when all constituent sources are consistently
miscalibrated in the same direction, as shown in
Sec.~\ref{sec: ptm general}. However,
Eq.~\eqref{eq: c_eff directionality} concerns only the
coefficient of the identity trajectory. It does not, by
itself, determine the sign of the complete DFE overlap in
Eq.~\eqref{eq: fidelity pauli coefficients}, since the
remaining trajectory coefficients can have either sign
and their contributions can accumulate. We therefore use
$c_{\mathrm{eff}}(\vec I)$ only as motivation and
establish the local response of the complete overlap
through the perturbative result below.

\begin{lemma}[Local overestimation response beyond Clifford circuits]
Consider a fixed arbitrary unitary circuit with a finite
number $s$ of layers acting on the arbitrary pure initial
state $\rho_0$. Let
\begin{equation}
    \rho_k
    =
    G_k\cdots G_1\rho_0
    G_1^\dagger\cdots G_k^\dagger
\end{equation}
denote the ideal state after layer $k$, with ideal final
state $\rho=\rho_s$.

Let $\lambda\geq0$ parameterize the magnitude of the
miscalibration along a fixed overestimation direction,
with $\lambda=0$ corresponding to perfect calibration.
Assume that the residual maps are trace preserving and
Pauli diagonal and that their effective Pauli fidelities
are right differentiable at $\lambda=0$, with
\begin{equation}
    f_{k,P,\mathrm{eff}}(\lambda)
    =
    1+\lambda r_{k,P}
    +o(\lambda),
    \label{eq: perturbative effective fidelities}
\end{equation}
where
\begin{equation}
    r_{k,P}
    =
    \left.
    \frac{\mathrm{d}}{\mathrm{d}\lambda}
    f_{k,P,\mathrm{eff}}(\lambda)
    \right|_{\lambda=0^+}.
\end{equation}
Assume further that
\begin{equation}
    f_{k,P,\mathrm{eff}}(\lambda)
    \geq
    1
    \qquad
    \forall\,k,P
\end{equation}
for sufficiently small positive $\lambda$.

If there exists at least one pair $(k,P)$ such that
\begin{equation}
    r_{k,P}>0
    \qquad\text{and}\qquad
    \operatorname{Tr}(P\rho_k)\neq0,
    \label{eq: strict perturbative condition}
\end{equation}
then there exists a $\lambda_0>0$ such that
\begin{equation}
    \operatorname{Tr}
    \!\left[
        \widetilde{\rho}(\lambda)\rho
    \right]
    >
    1
    \qquad
    \forall\,0<\lambda<\lambda_0.
\end{equation}
\end{lemma}

\begin{proof}
Since the Pauli basis is finite, the corresponding
effective map has the expansion
\begin{equation}
    \mathcal{D}_{k,\mathrm{eff}}^{(\lambda)}
    =
    \mathcal{I}
    +\lambda\mathcal{L}_k
    +o(\lambda),
\end{equation}
where
\begin{equation}
    \mathcal{L}_k(P)
    =
    r_{k,P}P.
\end{equation}
Trace preservation gives $r_{k,I}=0$. Along a
consistent-overestimation direction, the effective
fidelities satisfy
$f_{k,P,\mathrm{eff}}(\lambda)\geq1$ for sufficiently
small positive $\lambda$. Therefore, their right
derivatives obey
\begin{equation}
    r_{k,P}\geq0
    \qquad
    \forall\,k,P.
    \label{eq: perturbative overestimation condition}
\end{equation}

For each $k$, define the suffix unitary
\begin{equation}
    V_k
    =
    G_sG_{s-1}\cdots G_{k+1},
\end{equation}
with $V_s=I$. Applying the remaining ideal layers to
$\rho_k$ gives the final ideal state:
\begin{equation}
    \rho
    =
    V_k\rho_kV_k^\dagger
    \qquad
    \forall\,k\in\{1,\ldots,s\}.
    \label{eq: ideal suffix evolution}
\end{equation}
For example, $V_s=I$ gives $\rho=\rho_s$, while
$V_{s-1}=G_s$ gives
$\rho=G_s\rho_{s-1}G_s^\dagger$.

For a fixed circuit with a finite number of layers,
expanding the complete effective mitigated circuit to
first order in $\lambda$ gives
\begin{equation}
    \widetilde{\rho}(\lambda)
    =
    \rho
    +
    \lambda
    \sum_{k=1}^{s}
    V_k\mathcal{L}_k(\rho_k)V_k^\dagger
    +
    o(\lambda).
    \label{eq: perturbative mitigated state}
\end{equation}
Each term in the sum corresponds to taking the
first-order contribution from the residual map at layer
$k$ and the zeroth-order identity contribution from the
residual maps at all other layers. Contributions involving
first-order deviations at two or more layers are of higher
order and are included in the remainder.

Taking the overlap with the ideal state and using
Eq.~\eqref{eq: ideal suffix evolution}, cyclicity of the
trace, and $\operatorname{Tr}(\rho^2)=1$, yields
\begin{align}
    \operatorname{Tr}\!\left[
        \widetilde{\rho}(\lambda)\rho
    \right]
    &=
    1
    +
    \lambda
    \sum_{k=1}^{s}
    \operatorname{Tr}\!\left[
        \rho_k\mathcal{L}_k(\rho_k)
    \right]
    +
    o(\lambda).
    \label{eq: perturbative overlap intermediate}
\end{align}
Using the Pauli expansion
\begin{equation}
    \rho_k
    =
    \frac{1}{2^N}
    \sum_{P\in\mathcal{P}_N}
    \operatorname{Tr}(P\rho_k)\,P
\end{equation}
and the diagonal action
$\mathcal{L}_k(P)=r_{k,P}P$, we obtain
\begin{equation}
    \operatorname{Tr}\!\left[
        \rho_k\mathcal{L}_k(\rho_k)
    \right]
    =
    \frac{1}{2^N}
    \sum_{P\in\mathcal{P}_N}
    r_{k,P}\,
    \operatorname{Tr}(P\rho_k)^2.
\end{equation}
Therefore,
\begin{equation}
    \operatorname{Tr}\!\left[
        \widetilde{\rho}(\lambda)\rho
    \right]
    =
    1+\lambda A+o(\lambda),
    \label{eq: perturbative DFE response}
\end{equation}
with
\begin{equation}
    A
    =
    \frac{1}{2^N}
    \sum_{k=1}^{s}
    \sum_{P\in\mathcal{P}_N}
    r_{k,P}\,
    \operatorname{Tr}(P\rho_k)^2.
    \label{eq:perturbative-DFE-coefficient-appendix}
\end{equation}

In the consistent-overestimation case, Eq.~\eqref{eq: perturbative overestimation condition} implies $A\geq0$, since every term in Eq.~\eqref{eq:perturbative-DFE-coefficient-appendix} is nonnegative. Under the additional condition \eqref{eq: strict perturbative condition}, we have $A>0$. By the definition of $o(\lambda)$, there exists a $\lambda_0>0$ such that
\begin{equation}
    \left|o(\lambda)\right|
    \leq
    \frac{\lambda A}{2}
    \qquad
    \forall\,0<\lambda<\lambda_0.
\end{equation}
It follows that
\begin{equation}
    \operatorname{Tr}\!\left[
        \widetilde{\rho}(\lambda)\rho
    \right]
    \geq
    1+\frac{\lambda A}{2}
    >
    1
    \qquad
    \forall\,0<\lambda<\lambda_0.
\end{equation}
\end{proof}

If $A=0$, the first-order correction vanishes and the perturbative expansion does not determine the sign of the DFE overlap in the overestimation case. Higher-order terms must then be considered. \section{Lower bounding the DFE purity in the overestimation case}
\label{appendix: purity proof}

The effective-fidelity condition in the lemma below holds for an individually overestimated error source and is preserved when several constituent error sources are consistently overestimated.

\begin{lemma}[DFE-purity bound under overestimation]
Consider an arbitrary unitary circuit acting on a pure initial state. Assume that every residual map is Pauli diagonal and that its effective Pauli fidelities satisfy
\begin{equation}
    f_{k,P,\mathrm{eff}}
    \geq
    1
    \qquad
    \forall\,k,P.
\end{equation}
Then the DFE purity of the effective mitigated state cannot decrease between consecutive layers and the final effective mitigated state satisfies
\begin{equation}
    \mathrm{Tr}(\widetilde{\rho}^2)
    \geq
    1.
\end{equation}
\end{lemma}

\begin{proof}
We begin by considering the first layer. The DFE purity of the effective mitigated state after the first mitigated layer is
\begin{align}
    \mathrm{Tr}(\widetilde{\rho}_1^2)
    &=
    \frac{1}{2^N}
    \sum_{P\in\mathcal{P}_N}
    f_{1,P,\mathrm{eff}}^2\,
    \mathrm{Tr}(P\rho_1)^2
    \notag\\
    &\geq
    \frac{1}{2^N}
    \sum_{P\in\mathcal{P}_N}
    \mathrm{Tr}(P\rho_1)^2
    \notag\\
    &=
    \mathrm{Tr}(\rho_1^2)
    =
    1.
\end{align}

Next, we show that the DFE purity cannot decrease after each subsequent layer satisfying the effective-fidelity condition. As in the Clifford proof, we define the fictitious state
\begin{equation}
    \sigma_{k+1}
    =
    G_{k+1}\widetilde{\rho}_kG_{k+1}^\dagger,
\end{equation}
which corresponds to the effective mitigated state after applying the $(k+1)$-th noiseless gate but before applying the $(k+1)$-th effective mitigated map. By cyclicity of the trace,
\begin{equation}
    \mathrm{Tr}(\sigma_{k+1}^2)
    =
    \mathrm{Tr}(\widetilde{\rho}_k^2).
\end{equation}
Then, using the diagonal action of the effective Pauli-diagonal map in the PTM representation,
\begin{align}
    \mathrm{Tr}(\widetilde{\rho}_{k+1}^2)
    &=
    \frac{1}{2^N}
    \sum_{P\in\mathcal{P}_N}
    f_{k+1,P,\mathrm{eff}}^2\,
    \mathrm{Tr}(P\sigma_{k+1})^2
    \notag\\
    &\geq
    \frac{1}{2^N}
    \sum_{P\in\mathcal{P}_N}
    \mathrm{Tr}(P\sigma_{k+1})^2
    \notag\\
    &=
    \mathrm{Tr}(\sigma_{k+1}^2)
    =
    \mathrm{Tr}(\widetilde{\rho}_k^2).
\end{align}
We have therefore proven that the DFE purity of the effective
mitigated state cannot decrease after any layer satisfying
the effective-fidelity condition. Since after the first
layer the DFE purity already satisfies
\begin{equation}
    \mathrm{Tr}(\widetilde{\rho}_1^2)
    \geq
    1,
\end{equation}
it follows that
\begin{equation}
    \mathrm{Tr}(\widetilde{\rho}^2)
    \geq
    1
\end{equation}
for arbitrary unitary circuits whenever $f_{k,P,\mathrm{eff}}\geq1$ for every $k$ and $P$. In particular, this includes concatenated noise models for which all constituent error sources are consistently overestimated.
\end{proof}
\section{Extended numerical results for random Matchgate circuits}
\label{appendix:extended_MG_numerics}

To assess the robustness of the numerical observations across different circuit structures and input states, we extend the analysis to an ensemble of ten independently generated Matchgate circuit instances. For each instance, the circuit $U$ contains eight two-qubit Matchgates of the form
\begin{equation}
    M(\beta_k,\tau_k)
    =
    e^{i\beta_k X_iX_j}e^{i\tau_k Y_iY_j},
\end{equation}
with the angles and the pairs of neighboring qubits chosen randomly, as in the main-text example. The protocol is applied to the corresponding $UU^\dagger$ circuit acting on $N=12$ qubits.

As in the main-text numerical study, all simulations use standard gate-local PEC. None of the sampling-reduction strategies discussed in Sec.~\ref{sec: reducing PEC overhead} are applied in producing the results below.

For the first circuit instance, the input is fixed to the all-zero state $|0\rangle^{\otimes N}$. For each of the remaining nine circuit instances, the input is chosen independently and uniformly at random from the computational-basis product states. Equivalently, each qubit is initialized independently in either $|0\rangle$ or $|1\rangle$ with equal probability. The sampled input state is then kept fixed for all parameter values and repetitions associated with that circuit instance. Since the ideal $UU^\dagger$ circuit returns its input state, every ideal output is a computational-basis product state and is therefore well conditioned with $\alpha=1$.

For every circuit, the true noise parameters are fixed to
\begin{equation}
    p_c=0.005,
    \qquad
    \gamma=0.05.
\end{equation}
The assumed parameters are varied over
\begin{equation}
    \widetilde p_c\in[0,0.015],
    \qquad
    \widetilde\gamma\in[0,0.15],
\end{equation}
using increments $\Delta p_c=0.001$ and
$\Delta\gamma=0.01$. At every point
$(\widetilde p_c,\widetilde\gamma)$, the DFE protocol is repeated ten
independent times using $\epsilon=0.05$ and $\delta=0.1$.

Figure~\ref{fig:extended_numerics_acceptance} shows the empirical acceptance frequency of the overlap test in Eq.~\eqref{eq: protocol}. Each panel corresponds to a different random MG circuit and its associated computational-basis input state. In every instance, the acceptance region is concentrated around the correctly calibrated parameters, marked by the white cross. The curved shape of this region reflects the partial compensation between an overestimation of one noise parameter and an underestimation of the other. Although the precise boundary depends on the circuit instance, the same qualitative acceptance structure is observed for all ten circuits.

\begin{figure*}[t]
    \centering
    \includegraphics[width=\textwidth]
    {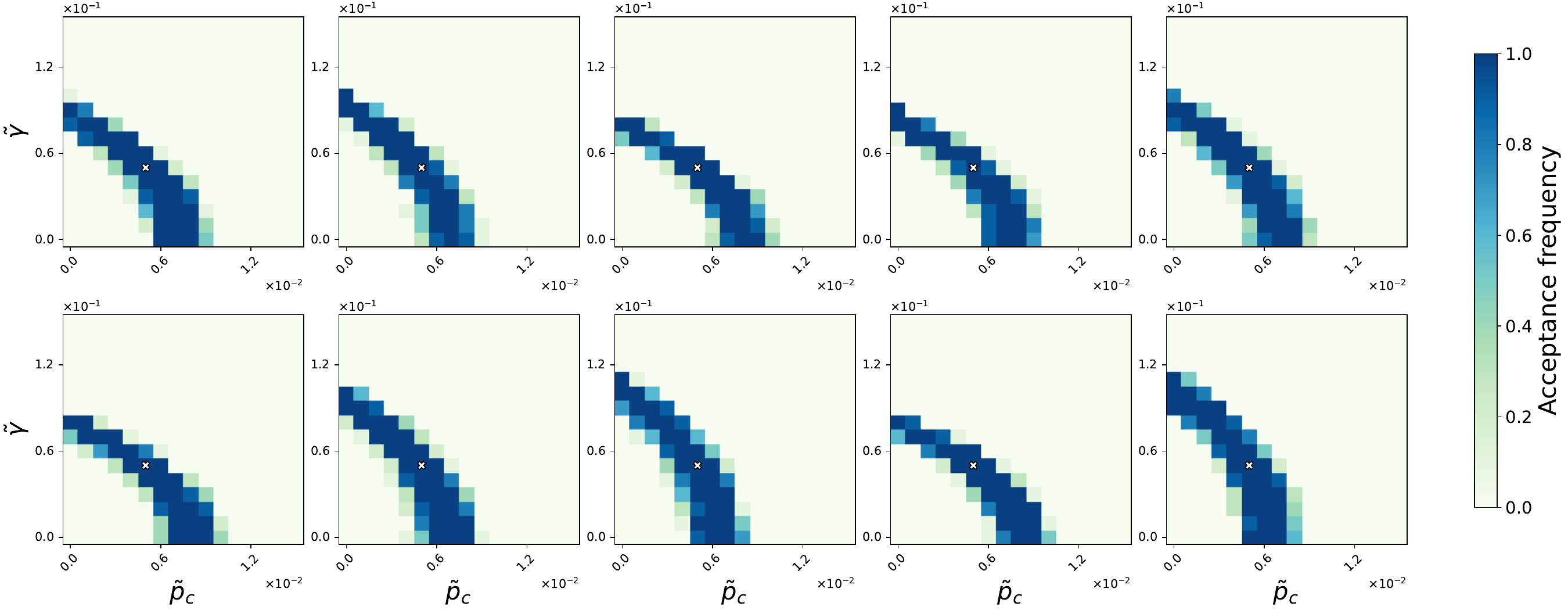}
    \caption{
    Acceptance frequency of the DFE-overlap test for ten independently generated MG circuits. The panels, ordered from left to right in the top row and then in the bottom row, correspond to the ten circuit instances. Each circuit acts on $N=12$ qubits and $U$ contains eight randomly chosen two-qubit MGs. The first circuit uses the all-zero input state, while each of the remaining nine circuits uses an independently sampled random computational-basis product state that is kept fixed throughout the corresponding numerical experiment. For each pair of assumed parameters $(\widetilde p_c,\widetilde\gamma)$, the acceptance frequency is the fraction of ten independent repetitions satisfying $|\widehat{Y}-1|\leq2\epsilon$, with $\epsilon=0.05$ and $\delta=0.1$. The white cross marks the true parameters, $p_c=0.005$ and $\gamma=0.05$.}
    \label{fig:extended_numerics_acceptance}
\end{figure*}

We also examine whether the direction of the DFE-overlap deviation is reproduced across the circuit ensemble. For each repetition, we assign the indicator
\begin{equation}
    d(\widehat{Y})=
    \begin{cases}
        1,   & \widehat{Y}>1+2\epsilon,\\
        1/2, & |\widehat{Y}-1|\leq2\epsilon,\\
        0,   & \widehat{Y}<1-2\epsilon.
    \end{cases}
\end{equation}
The value plotted in Fig.~\ref{fig:extended_numerics_direction} is the average of this indicator over the ten independent repetitions.

Across all circuit instances, consistent underestimation places the DFE overlap below one, whereas consistent overestimation places it above one, outside the statistically inconclusive region. Hybrid miscalibration produces a circuit-dependent transition band in which the two parameter deviations can partially compensate.

\begin{figure*}[t]
    \centering
    \includegraphics[width=\textwidth]
    {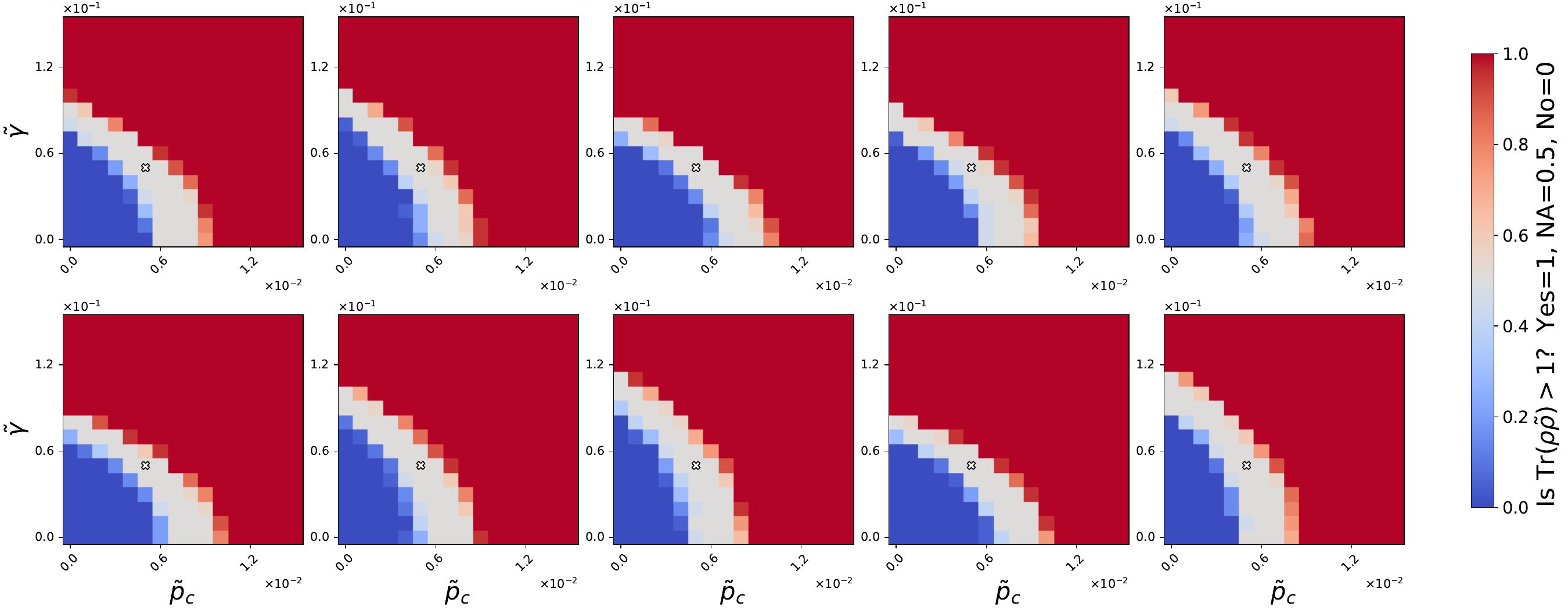}
    \caption{
    Direction of the DFE-overlap deviation for the same ten random MG circuits and computational-basis input states as in Fig.~\ref{fig:extended_numerics_acceptance}. For each pair $(\widetilde p_c,\widetilde\gamma)$, the color represents the indicator $d(\widehat{Y})$, averaged over ten independent repetitions: $d=1$ indicates an overlap statistically above one, $d=0$ indicates an overlap statistically below one, and $d=1/2$ indicates an inconclusive result within $1\pm2\epsilon$. We use $\epsilon=0.05$ and $\delta=0.1$. The white cross marks the true parameters, $p_c=0.005$ and $\gamma=0.05$.}
    \label{fig:extended_numerics_direction}
\end{figure*}

\end{document}